\documentclass[11pt,a4paper,reqno]{article}
\usepackage{amssymb}
\usepackage{latexsym}
\usepackage{amsmath}
\usepackage{graphicx}
\usepackage{relsize}
\usepackage{amsthm}
\usepackage{empheq}
\usepackage{bm}
\usepackage{booktabs}
\usepackage[dvipsnames]{xcolor}
\usepackage{pagecolor}
\usepackage{subcaption}
\usepackage{tikz,pgfplots,lipsum,lmodern}
\usepackage{enumitem}
\usepackage[most]{tcolorbox}
\newcommand\binomi[2]{\left(\begin{matrix} #1 \\ #2 \end{matrix} \right)}

\newcommand{\dstirling}[2]{\genfrac{[}{]}{0pt}{0}{#1}{#2}}
\usepackage[margin=2.5cm]{geometry}

 \usepackage[section]{placeins}

\numberwithin{equation}{section}

\theoremstyle{definition}
\newtheorem{theorem}{Theorem}[section]
\newtheorem{corollary}[theorem]{Corollary}

\newtheorem{definition}[theorem]{Definition}

\newtheorem{notation}[theorem]{Notation}
\newtheorem{remark}[theorem]{Remark}
\newtheorem{lemma}[theorem]{Lemma}

\usepackage{calrsfs}

\newcommand\qbin[3]{\left[\begin{matrix} #1 \\ #2 \end{matrix}\right]_{#3}}

\newcommand\bbqrk[1]{\textbf{v}_{q}^{\textnormal{rk}}(#1)}

\newcommand\bbqsr[1]{\textbf{v}_{q}^{\textnormal{sr,t}}(#1)}

\newcommand{\numberset}{\mathbb}

\newcommand{\R}{\numberset{R}}

\newcommand{\F}{\numberset{F}}

\newcommand{\mV}{\mathcal{V}}

\newcommand{\mC}{\mathcal{C}}

\newcommand{\mF}{\mathcal{F}}
\newcommand{\mW}{\mathcal{W}}

\newcommand{\mB}{\mathcal{B}}
\newcommand{\mE}{\mathcal{E}}

\newcommand{\rk}{\textnormal{rk}}

\newcommand{\matt}{\F_{q^{\ell}}^{n \times s}}
\renewcommand{\longrightarrow}{\to}

\newcommand{\dH}{D^{\textnormal{H}}}
\newcommand{\wH}{\omega^{\textnormal{H}}}

\newcommand{\wsr}{\omega^{\textnormal{sr,t}}}
\newcommand{\drk}{D^{\textnormal{rk}}}
\newcommand{\dsr}{D^{\textnormal{sr,t}}}

\newcommand{\wrk}{\omega^{\rk}}

\newcommand*{\myproofname}{Proof of the claim}

\makeatletter
\renewcommand*\env@matrix[1][*\c@MaxMatrixCols c]{%
  \hskip -\arraycolsep
  \let\@ifnextchar\new@ifnextchar
  \array{#1}}
\makeatother

\usepackage{hyperref}

\title{
\textbf{Densities of Codes of Various Linearity Degrees \\ in Translation-Invariant Metric Spaces}}

\usepackage{authblk}

\author[1]{Anina Gruica\thanks{A. G. is supported by the Dutch Research Council through grant OCENW.KLEIN.539.}}
\affil[1]{Eindhoven University of Technology, the Netherlands}

\author[2]{Anna-Lena Horlemann}
\affil[2]{University of St.Gallen, Switzerland}

\author[1]{Alberto Ravagnani\thanks{A. R. is supported by the Dutch Research Council through grants VI.Vidi.203.045 and 
OCENW.KLEIN.539,
by the European Commission through MSCA-DN grant
``European Network in Coding Theory and Applications'',
and by the Royal Academy of Arts and Sciences of the Netherlands.}}
\author[2]{Nadja Willenborg}

\date{}

\usepackage{setspace}

\begin{document}

\maketitle
	
\thispagestyle{empty}
	
\begin{abstract}
We investigate the asymptotic 
density of error-correcting codes
with good distance properties
and prescribed linearity degree, including (sub)linear and nonlinear codes. We focus on the general setting of finite translation-invariant metric spaces,
and then specialize our results to the Hamming metric, to the rank metric, and to the sum-rank metric.
Our results show that
the asymptotic density of 
codes heavily depends on the imposed linearity degree and the chosen metric. 
\end{abstract}

\bigskip
  
\bigskip

\section*{Introduction}
The asymptotic performance of error-correcting codes is a classical topic in coding theory, going back to the work of Shannon, where random codes are used to get arbitrarily close to the capacity of a given discrete memoryless channel \cite{shannon48}.
As in Shannon's proof, most of the literature has focused 
on the case where the field size is fixed and the code length goes to infinity.
Recently, increasing attention has been paid to the case where the field size grows and the other code parameters stay fixed, especially in connection with open questions in the theory of rank-metric codes and their related open questions in semifield theory~\cite{gruica2022rank}. For example, it was unknown if there are maximum rank distance codes that are not equivalent to a Gabidulin code, until it was shown that maximum rank distance codes are dense in the set of linear codes (for growing field extension degree of the ambient space), whereas Gabidulin codes are not~\cite{neri2018genericity}. Furthermore, lower bounds (in particular the Gilbert-Varshamov bound) for the minimum distance of a randomly chosen code (attained by high probability) are needed in several applications in code-based cryptography, see e.g. \cite{baldi2021restricted,CVE,Vron2009ImprovedIS}. Such lower bounds can again be derived by density results about the respective code families.

In this paper, we investigate the proportion, or density, of error correcting codes with good distance properties, distinguishing between different degrees of linearity of the codes. We develop a general framework for determining (upper and lower bounds on) densities of codes in discrete 
translation-invariant metric spaces, with a focus on the asymptotic behavior of these densities. We then use this framework to determine the asymptotic densities of linear, (sub)linear and nonlinear codes over finite fields with respect to various metrics, namely the Hamming metric, the rank metric, and the sum-rank metric.
In particular, we study extremal (or optimal) codes, in the sense that they attain the Singleton-type bound for the respective metric. 

Considering codes in $\F_{q^{m}}^{n}$ endowed with the Hamming metric it is known that if $n < q^{m}+1$ there exist linear maximum distance separable (MDS) codes, i.e., codes that achieve the classical Singleton bound, e.g., Reed-Solomon codes \cite{reed1960polynomial} and their generalizations. Furthermore, it is 
folklore knowledge that linear MDS codes are dense as the field size tends to infinity. This follows, for example,
by the fact that all the maximal minors 
of a uniformly random
matrix over $\F_q$ are nonzero with probability approaching 1 as $q$ goes to infinity. 
Only few results are known about the nonlinear case, e.g., in \cite{barg2002random} it was observed that over a binary field the minimum distance of random nonlinear codes is asymptotically worse than in the linear setting. They even show that for growing $n$ that the order is square (which corresponds to our square in the little-o notation. Nevertheless they neither considered this with respect to the Singleton bound nor asymptotically as the field size tends to infinity. In \cite{gruica2021typical}, the sparsity of nonlinear MDS codes over $\F_q$ was established for $q$ large.
Furthermore, the question of the existence of (sub)linear MDS codes has recently received attention, due to their application in quantum coding. For example in \cite{ball2020additive} a geometric approach is used to classify all additive codes over $\F_{9}$. Then the MDS conjecture for this class of codes could be verified and hence the quantum MDS conjecture over~$\F_{3}$. 

In the rank metric it is also known that linear maximum rank-distance (MRD) codes in~$\F_{q^m}^n$ exist  if $n \le m$, namely Gabidulin codes \cite{gabidulin}. Moreover, 
several density results of MRD codes are known, depending on whether they are nonlinear,  $\F_{q}$-linear or even linear over the extension field~$\F_{q^m}$. In particular, it has been shown in \cite{neri2018genericity} that~$\F_{q^m}$-linear MRD codes are dense as $m \to +\infty$. But if the codes belong to the wider class of~$\F_{q}$-linear MRD codes they are not dense, neither for $q \to +\infty$, nor for $m \to + \infty$ (except for trivial parameter choices). In particular, $\F_{q}$-linear MRD codes are sparse as $q \to +\infty$ \cite{gruica2022common}. 

In the sum-rank metric the asymptotic densities for an arbitrary degree of linearity remained mostly unexplored. Recently, in \cite{ott2021bounds}, the Schwartz-Zippel Lemma was used to show that for $m \rightarrow + \infty$, most $\mathbb{F}_{q^{m}}$-linear codes in $\mathbb{F}_{q^{m}}^{n}$ are maximum sum-rank distance (MSRD) codes.     

This contrasting behavior with respect to different metrics and different types of linearity motivated us to study the density of codes in general metric spaces, with various degrees of linearity. 

For this, we consider codes in $\F_{q^m}^n$, equipped with an arbitrary translation-invariant metric~$D$.\footnote{The results can easily be extended for more general finite metric spaces, where the size of the balls of radius~$r$ is independent of the center. For simplicity, however, we will restrict ourselves to $\F_{q^m}^n$ and translation-invariant metrics.} For the linearity degree of the code we consider the maximal subfield $\F_{q^{\ell}} \subseteq\F_{q^m}$ over which the codes are linear. Clearly, $\ell$ is a divisor of $m$. 

We then classify the asymptotic densities of (possibly) nonlinear codes and codes that are linear over the subfield $\F_{q^{\ell}} \subseteq \F_{q^m}$ separately.
We also treat the asymptotic density with respect to the four parameters $q,n,\ell$ and $s := [\F_{q^m} \colon \F_{q^{\ell}}]$ separately. 

\paragraph*{Our contribution.}   We determine upper and lower bounds on the proportion of codes in~$\F_{q^m}^n$ with a prescribed minimum distance (with respect to an arbitrary translation-invariant distance~$D$) among all codes of a fixed cardinality. For this, we distinguish the case of nonlinear and~$\F_{q^\ell}$-linear codes in~$\F_{q^m}^n$, for any suitable $\ell$. We then analyze the asymptotic behavior of these bounds and show that this highly depends on the volume of the balls with respect to $D$. In particular, we give conditions on the volume with respect to the prescribed cardinality and minimum distance, such that the family of codes is dense or sparse (both in the nonlinear or (sub)linear case). With these results we can straightforwardly see that nonlinear random codes achieve the Gilbert-Varshamov bound with probability going to zero, both for growing length or field size, while (sub)linear codes achieve it with probability going to $1$, for growing field size or linearity degree (but not for growing length).

In the Hamming metric, we show that independent of the linearity degree, the class of $\F_{q^{\ell}}$-linear MDS codes, as well as the class of nonlinear MDS codes, are both sparse as $n \to +\infty$.

When considering the asymptotic densities with respect to the field size, (sub)linear MDS codes are dense as $q \to +\infty$ (or $\ell \to +\infty$). On the other hand, the probability that a nonlinear code is MDS tends to zero as $q \to +\infty$. Finally, we give an upper bound for the asymptotic density of MDS codes with respect to $s$ which shows that MDS codes are not dense as $s \to +\infty$.
 
In the rank metric we show that nonlinear MRD codes are sparse with respect to both the field size and the code length. Moreover, we recover the previously known results on the sparsity and density of $\F_q$- or $\F_{q^m}$-linear MRD codes. As new contributions we derive more general results for $\F_{q^\ell}$-linear MRD codes with respect to any of the parameters $q,n,\ell, s$. In particular, we derive bounds on $\ell$, for which MRD codes are sparse and for which they are dense, for $q\rightarrow +\infty$. For $\ell \rightarrow +\infty$ we show that MRD codes are always dense.  For growing~$n$ or~$s$ we derive again upper bounds on the density.

In the sum-rank metric we show that nonlinear MSRD codes are sparse with respect to the field size. For (sub)linear MSRD codes and letting $q$ go to infinity, we derive bounds on $\ell$ such that the parameters for which MSRD codes are sparse and for which they are dense can be classified.

\vspace{0.3cm}
\textbf{The outline} of this article is as follows: In Section 2 we start by introducing the relevant terminology used throughout the article. We define the density functions associated with families of (sub)linear as well as (possibly) nonlinear codes. We also clarify the terms sparsity and density with respect to the parameters $q,n,s$ and $\ell$. Then we briefly provide the graph theory tools to obtain the upper and lower bounds on the density of a code family within a larger family. For more details on this and the specific bipartite graphs we refer to \cite{gruica2022common}.
In Section 3 we consider the asymptotic versions of these bounds with respect to the parameters $n,s, \ell$ and $q$. Under certain assumptions we can derive that the codes are sparse or dense with respect to a given parameter.
In the second part of this article, covering Sections 4 to 6, we apply the developed classifiers from Section 3 to codes with respect to various distance functions. Particularly, we will focus on three major distance functions, namely the Hamming, rank, and sum-rank distance.

\section{Preliminaries} \label{sec:prelim}

\subsection{Finite Metric Spaces and Codes}

Throughout this paper we let $n$ and $m$ be integers with $1 \le m,n$. Moreover, we let $q$ be a prime power, we denote the finite field of order $q$ by $\F_q$ and its extension field of degree $m$ by $\F_{q^m}$. From now on, unless otherwise stated, we work in $\F_{q^m}^n$ and consider a \emph{translation-invariant distance} $D: \F_{q^m}^n \times \F_{q^m}^n \rightarrow \R_{\ge 0}$.

\begin{notation} \label{not:landau}
We will repeatedly use the Bachmann-Landau notation (``Big O'', ``Little O'',``$\sim$'' and~``Omega'') to describe the asymptotic growth rate of functions defined on an infinite set of natural numbers; see for example~\cite{de1981asymptotic}. Since we are often interested in the asymptotics as the field size $q$ tends to infinity, we denote the set of prime powers by $Q$ and we omit ``$q \in Q$'' when writing $q \to +\infty$.
\end{notation}

\begin{definition} \label{def:code}
A subset $\mathcal{C} \subseteq \F_{q^m}^n$ is called a \emph{code}. If $2 \le |\mC|$ then the \emph{minimum distance} of a code $\mC \subseteq \F_{q^m}^n$ with respect to $D$ is
\begin{align*}
    D(\mC) := \min\{D(x,y) : x,y \in \mC, \, x \ne y\}.
\end{align*}
For a divisor $\ell$ of $m$, we denote an $\F_{q^\ell}$-linear code $\mC \subseteq \F_{q^m}^n$ of cardinality $S$ and minimum distance at least $d$ by $[\F_{q^m}^n,S,\ell,d]^D$-code. If on the other hand $\mC \subseteq \F_{q^m}^n$ is nonlinear, of cardinality $S$ and minimum distance at least $d$, we say $\mC$ is an $[\F_{q^m}^n,S,0,d]^D$-code.
\end{definition}

We will make extensive use of the $q$-ary binomial coefficient, which is defined as
$$\dstirling{a}{b}_{q} := \prod_{i=0}^{b-1}\frac{q^{ a}-q^{ i}}{q^{ b}-q^{ i}}$$ 
and counts the number of $b$-dimensional subspaces of an $a$-dimensional vector space over $\F_{q}$. Note that the number of $\F_{q^\ell}$-linear codes in $\F_{q^m}^n$ (for a divisor $\ell$ of $m$) of $\F_{q^\ell}$-dimension~$k$ is 
\begin{align*}
    |\{\mC \subseteq \F_{q^{m}}^n : |\mC|=q^{k \ell}, \, \textnormal{$\mC$ is $\F_{q^{\ell}}$-linear}\}| = \qbin{ns}{k}{q^\ell}.
\end{align*}
This formula will play a crucial role in determining bounds for the proportion of codes with certain distance properties among the set of codes of the same cardinality or dimension. 

\begin{definition}
 We let 
\begin{equation} \label{eq:density_nonlinear}
    \delta_q^D(\F_{q^m}^n,S,0,d) := \frac{|\{\mC \subseteq \F_{q^m}^n : |\mC|=S, \, D(\mC) \ge d\}|}{|\{\mC \subseteq \F_{q^m}^n : |\mC|=S\}|}
\end{equation}
denote the \emph{density of (possibly) nonlinear codes} of minimum distance $D(\mC)$ at least $d$ within the set of (possibly) nonlinear codes of the same cardinality $S$.
Analogously, 

we define 
\begin{equation} \label{eq:density_(sub)linear}
    \delta_q^D(\F_{q^m}^n,S,\ell,d) := \frac{|\{\mC \subseteq \F_{q^{m}}^n  :  |\mC|=S, \, D(\mC) \ge d, \, \textnormal{$\mC$ is $\F_{q^{\ell}}$-linear }\}|}{|\{\mC \subseteq \F_{q^{m}}^n  :  |\mC|=S, \, \textnormal{$\mC$ is $\F_{q^{\ell}}$-linear}\}|}
\end{equation}
as the \emph{density of $\F_{q^{\ell}}$-linear codes} of minimum distance $D(\mC)$ at least $d$ within the set of $\smash{\F_{q^{\ell}}}$-linear codes of the same cardinality $S$.

Taking the limit in Equation~\eqref{eq:density_nonlinear} or Equation~\eqref{eq:density_(sub)linear} as $q,n,s$ or $\ell$ tends to infinity (and the other three parameters are fixed constants) defines the \emph{asymptotic density}. 

If the limit is $0$ we say that $\F_{q^{\ell}}$-linear codes of minimum distance at least $d$ and cardinality $|S|$ are \emph{sparse} as $q,n,s$ or $\ell \to +\infty$. If the limit is~$1$ we say that such codes are \emph{dense} as $q,n,s$ or $\ell \to +\infty$.
\end{definition}

Since we assume $D$ to be invariant under translation, the ball with center $c \in \F_{q^m}^n$ of radius $0 \le r < \infty$ in $\F_{q^m}^n$, which is the set $\{x \in \F_{q^m}^n  :  D(x,c) \le r\}$, has the same size for any center $c \in \F_{q^m}^n$. Thus, when we consider the volume of the ball of radius $r$ in $\F_{q^m}^n$, we do not specify the center of the ball. For ease of exhibition, we propose the following notation.

\begin{notation} \label{def:ball}
We denote the volume of the ball of radius $0 \le r < + \infty$ in $\F_{q^m}^n$ by 
$$\textbf{v}_{q}^{D}(\F_{q^m}^n,r) := |\{x \in \F_{q^m}^n  :  D(x,0) \le r\}|.$$
\end{notation}

\subsection{Graph Theory Tools}
In this subsection we briefly recall some graph theory results from~\cite[Section 3]{gruica2022common}. As we will show shortly, studying the number of isolated vertices in certain bipartite graphs will give us bounds on the number of codes we are interested in.

\begin{definition}
A (\emph{directed}) \emph{bipartite graph} is a 3-tuple $\mB=(\mV,\mW,\mE)$, where $\mV$ and $\mW$ are finite non-empty sets and $\mE \subseteq \mV \times \mW$. The elements of $\mV \cup \mW$ are the \emph{vertices} of the graph and the tuples given by relation $\mE$ are called the \emph{edges} of $\mB$. We say that a vertex~$W \in \mW$ is
\emph{isolated} if there is no $V \in \mV$ such that $(V,W) \in \mE$. We say that the bipartite graph
$\mB$ is 
\emph{left-regular of degree $0 \le \partial$} if for all $V \in \mV$
$$|\{W \in \mW  :  (V,W) \in \mE\}| = \partial.$$
\end{definition}

To derive a lower bound for the number of non-isolated vertices in a left-regular bipartite graph, we introduce the concept of an association. This notion generalizes strong regularity properties on a bipartite graph with respect to certain functions defined on its left-vertices.

\begin{definition} \label{def:assoc}
Let $\mV$ be a finite non-empty set and let $0 \le r$ be an integer. An \emph{association} on $\mV$ of \emph{magnitude} $r$ is a function 
$\alpha: \mV \times \mV \to \{0,...,r\}$ satisfying the following:
\begin{itemize}
\item[(i)] $\alpha(V,V)=r$ for all $V\in \mV$;
\item[(ii)] $\alpha(V,V')=\alpha(V',V)$ for all $V,V' \in \mV$.
\end{itemize}
\end{definition}

\begin{definition}
Let $\mB=(\mV,\mW,\mE)$ be a finite bipartite graph and let $\alpha$ be an association on~$\mV$ of magnitude $r$.  We say that $\mB$ is \emph{$\alpha$-regular} if for all  $(V,V') \in \mV \times \mV$ the number of vertices $W \in \mW$ with $(V,W) \in \mE$ and 
$(V',W) \in \mE$ only depends on $\alpha(V,V')$. If this is the case, we denote this number by~$\mW_\ell(\alpha)$, where $\ell=\alpha(V,V') \in \{0, \dots , r \}$.
\end{definition}

\begin{remark}
Note that an $\alpha$-regular bipartite graph for an association $\alpha$ is necessarily left-regular of degree $\partial=\mW_r(\alpha)$.
\end{remark}

We can now bound the number of non-isolated vertices in a left-regular bipartite graph. With these two bounds we will derive a lower and an upper bound for the density function of codes in $\matt$ having certain distance properties.

 \begin{lemma}[\text{see \cite[Lemma 3.2]{gruica2022common}}] \label{lem:upperbound}
Let $\mB=(\mV,\mW,\mE)$ be a bipartite and left-regular graph of degree $0 < \partial$.
Let $\mF \subseteq \mW$ be the collection of non-isolated vertices of $\mW$.
We have
$$|\mF| \le |\mV| \, \partial.$$
\end{lemma}

The following lemma follows by combining the notion of an association and the Cauchy-Schwarz Inequality. 
\begin{lemma}[\text{see \cite[Lemma 3.5]{gruica2022common}}] \label{lem:lowerbound}
Let $\mB=(\mV,\mW,\mE)$ be a finite bipartite $\alpha$-regular graph, where $\alpha$ is an association on~$\mV$ of magnitude~$r$. Let $\mF \subseteq \mW$ be the collection of non-isolated vertices of $\mW$. If
$0 < \mW_r(\alpha)$, then 
$$\frac{\mW_r(\alpha)^2 \, |\mV|^2}{\sum_{\ell=0}^r  \mW_\ell(\alpha) \, |\alpha^{-1}(\ell)|} \le |\mF|.$$
\end{lemma}

\section{Upper and Lower Bounds on the Density of Codes}\label{sec:bounds}

In this subsection we use the upper bound of Lemma~\ref{lem:upperbound} and the lower bound of Lemma~\ref{lem:lowerbound} to give bounds on the number of codes in the metric space $(\F_{q^m}^n,D)$ that have minimum distance bounded from below by some positive integer $d$. From these bounds we will later obtain a prediction for the asymptotic behavior of the density of these codes.

\subsection{Bounds for Nonlinear Codes}
\label{sec:nonlinear}

We first consider nonlinear codes. Since we do not impose any linearity (nor (sub)linearity) on our codes in this subsection, we let $m=1$, meaning that our ambient space is $\F_{q}^n$ and as before,~$D$ is a translation-invariant distance on $\F_{q}^n$.

\begin{theorem} \label{thm:nonlinbound}
Let $2 \le S$ and $1 \le d < +\infty$ be integers. Define the quantities
\begin{align*}
    \beta^0 &= \frac{1}{2}q^n (\textbf{v}_q^{D}(\F_{q}^n,d-1)-1) - 2\textbf{v}_q^{D}(\F_{q}^n,d-1) +3, \\
    \beta^1 &= 2\textbf{v}_q^{D}(\F_{q}^n,d-1)-4, \\
    \Theta &= 1 + \beta^1  \frac{S-2}{q^n-2} \, + \, \beta^0  \frac{(S-2)(S-3)}{\left(q^n-2\right)\left(q^n-3\right)},
\end{align*}
and let $\mF:= \{\mC \subseteq \F_{q}^n  :  |\mC|=S, \, D(\mC) \le d-1\}$. We have
\begin{align*}
   \frac{1}{2\Theta} q^n \left(\textbf{v}_q^{D}(\F_{q}^n,d-1)-1\right) \binom{q^n-2}{S-2}\le  |\mF| \le \frac{1}{2} q^n \left( \textbf{v}_q^{D}(\F_{q}^n,d-1)-1\right)\binomi{q^n-2}{S-2}.
\end{align*}

 \begin{proof}
 We work with the bipartite graph $\mB=(\mV,\mW,\mE)$, where $$\mV=\{\{x,y\} \subseteq \F_{q}^n  :  x \ne y, \, D(x, y) \le d-1 \},$$ $\mW$ is the collection of codes $\mC \subseteq \F_{q}^n$ with $|\mC| = S$, and 
  $(\{x,y\},\mC) \in \mE$ if and only if $\{x,y\} \subseteq \mC$. 
Note that the set of non-isolated vertices in $\mW$ is exactly $\mF$. 
 We have \begin{align*}
     |\mV| = \frac{1}{2} q^n \left(\textbf{v}_q^{D}(\F_{q}^n,d-1)-1\right), \quad  |\mW| = \binomi{q^{n}}{S}.
 \end{align*}
Moreover, one easily checks that
 \begin{align*}
    |\{ \mC \in \mW  :  (\{x,y\}, \mC) \in \mE\}| = \displaystyle \binom{q^{n}-2}{S-2}.
\end{align*}
Hence $\mB$ is a left-regular graph of degree $$\binomi{q^n-2}{S-2}.$$ Applying Lemma~\ref{lem:upperbound} we obtain the upper bound on $|\mF|$.

To prove the lower bound, we let $\alpha: \mV \times \mV \longrightarrow \{0,1,2\}$
be defined by $$\alpha(\{x,y\},\{z,t\}) := 4-|\{x,y,z,t\}|$$
 for all $x,y,z,t \in \F_{q}^n.$
We have 
 \begin{align*}
 |\alpha^{-1}(2)| &= |\mV|, \\
 |\alpha^{-1}(1)| &= 2|\mV|(\textbf{v}_q^{D}(\F_{q}^n,d-1)-2), \\
 |\alpha^{-1}(0)| &= |\mV|(|\mV|-2\textbf{v}_q^{D}(\F_{q}^n,d-1)+3).
 \end{align*}
It is easy to see that $|\alpha^{-1}(2)| = |\mV|$.
The elements of the domain $\alpha^{-1}(1)$ can be obtained by choosing 
  $\{x,y\} \in \mV$ arbitrary
  and then $\{z,t\} \in \mV$ with 
  either $z=x$ or $z=y$ and $$t \in \{v \in \F_{q}^n  :  D(v,z) \le d-1\} \backslash \{x,y\}.$$  Therefore $$ |\alpha^{-1}(1)|=2|\mV|(\textbf{v}_q^{D}(\F_{q}^n,d-1)-2).$$
 To compute $|\alpha^{-1}(0)|$ we simply note that  $$|\mV|^2=|\alpha^{-1}(0)|+|\alpha^{-1}(1)|+|\alpha^{-1}(2)|.$$ Therefore the value of $|\alpha^{-1}(0)|$ 
 follows from the values of $|\alpha^{-1}(1)|$ and $|\alpha^{-1}(2)|$.
 Simple counting arguments imply that the bipartite graph~$\mB$ is regular with respect to $\alpha$. Therefore for $(\{x,y\},\{z,t\}) \in \mV \times \mV$ and $\ell=\alpha(\{x,y\},\{z,t\})$ we define
 \begin{align*} 
     \mW_{\ell}(\alpha) := |\{W \in \mW  :  \{x,y,t,z\} \subseteq W\}|  =  \binomi{q^{n}-4+\ell}{S-4+\ell}.
     \end{align*}
      We can now apply Lemma~\ref{lem:lowerbound} obtaining that $|\mF|$ is lower bounded by
 \begin{equation*}
 \frac{\mW_{2}(\alpha)^2 \, |\mV|^2}{|\alpha^{-1}(2)|\mW_{2}(\alpha) + |\alpha^{-1}(1)|\mW_{1}(\alpha) + |\alpha^{-1}(0)|\mW_{0}(\alpha)}.
 \end{equation*}
Finally, plugging in the formulas for $|\alpha^{-1}(0)|$, $|\alpha^{-1}(1)|$ and $|\alpha^{-1}(2)|$, and applying the identity
\begin{align} \label{eq:binomi}
\binom{a}{b} = \frac{a}{b} \binom{a-1}{b-1},
\end{align}
for all $1 \le b \le a$ yields the desired result.
\end{proof}
\end{theorem}

As a corollary we get the following bounds on the proportion of codes in $\F_q^{n}$ of minimum distance at least $d$, again using the identity in~\eqref{eq:binomi}.

\begin{corollary} \label{cor:nonlindensity}
Let $2 \le S$ be an integer and consider $1 \le d < +\infty$. We have
\begin{align*}
    1- \frac{(\textbf{v}_q^{D}(\F_q^n,d-1)-1)S(S-1)}{2\left(q^n-1\right)} \le \delta_q^D(\F_q^n,S,0,d) \le 1-\frac{(\textbf{v}_q^{D}(\F_q^n,d-1)-1)S(S-1)}{2\Theta(q^n-1)},
\end{align*}
where $\Theta$ is the same as in Theorem~\ref{thm:nonlinbound}.
\end{corollary}

\subsection{Bounds for (Sub)Linear Codes}

In this subsection, we fix a divisor $\ell$ of $m$ and let 
$s= [\mathbb{F}_{q^{m}}  :  \mathbb{F}_{q^{\ell}}]$.
We will consider $\F_{q^\ell}$-linear codes in the metric space $(\F_{q^m}^n,D)$.

We start by providing the bounds obtained using the tools from graph theory.

\begin{theorem}\label{thm:subavoid}
Let $1 \le k \le ns$ and $1 \le d < +\infty$ be integers. Let
\begin{align*}
    \mF := \{\mC \subseteq \F_{q^m}^n  :  |\mC|=q^{\ell k}, \, D(\mC) \le d-1, \, \textnormal{$\mC$ is $\F_{q^{\ell}}$-linear}\}.
\end{align*}
We have
\begin{align*}
    |\mF| &\le \left(\displaystyle\frac{\textbf{v}_q^{D}(\F_{q^m}^n,d-1)-1}{q^{\ell}-1}\right)\dstirling{ns-1}{k-1}_{q^{\ell}}, \\[0.2cm]
    |\mF| &\ge \frac{\displaystyle \frac{\textbf{v}_{q}^{D}(\F_{q^m}^n,d-1)-1}{q^{\ell}-1} \dstirling{ns-1}{k-1}_{q^{\ell}}^{2}}{\dstirling{ns-1}{k-1}_{q^{\ell}}+\left(\displaystyle \frac{\textbf{v}_q^{D}(\F_{q^m}^n,d-1)-1}{q^{\ell}-1} -1 \right)\dstirling{ns-2}{k-2}_{q^{\ell}}}.
\end{align*}
\end{theorem}
\begin{proof}
 We work with the bipartite graph $$\mB=(\mV,\mW,\mE),$$ where $\mV$ is the set of elements in $\F_{q^m}^n$ of minimum distance at most $d-1$ (up to multiples by elements of $\F_{q^\ell}$) from 0, $\mW$ is the collection of $\F_{q^\ell}$-linear codes in $\F_{q^m}^n$ of $\F_{q^\ell}$-dimension $k$, and $(M,\mC) \in \mE$ if and only if $M \in \mC$. Hence $\mF$ is the family of non-isolated vertices in $\mW$.
Note that we have
\begin{align*}
    |\mV| = \frac{\textbf{v}_q^{D}(\F_{q^m}^n,d-1)-1}{q^{\ell}-1}, \quad  |\mW| = \dstirling{ns}{k}_{q^{\ell}}.
\end{align*}
Moreover, for $M \in \mV$ we have
\begin{align*}
    |\{ \mC \in \mW  :  (M, \mC) \in \mE\}|=  \dstirling{ns-1}{k-1}_{q^{\ell}}.
\end{align*}
Hence the bipartite graph $\mB$ is left-regular of degree 
$\dstirling{ns-1}{k-1}_{q^{\ell}}.$
With this, the upper bound in the theorem is an easy consequence of Lemma~\ref{lem:upperbound}. 
For the lower bound, consider the association $$\alpha  :  \mV \times \mV \longrightarrow \{0,1\}, \quad (V,V') \mapsto 2-\dim\langle V,V' \rangle.$$ It is easy to see that 
\begin{align*}
       |\alpha^{-1}(0)|= |\mV| (|\mV|-1),\quad |\alpha^{-1}(1)|= |\mV|
\end{align*}
and $\mB$ is $\alpha$-regular. Furthermore we have
\begin{align*}
    \mW_0(\alpha)= \qbin{ns-2}{k-2}{q^\ell}, \; \mW_1(\alpha)= \qbin{ns-1}{k-1}{q^\ell}
\end{align*}
which combined with Lemma~\ref{lem:lowerbound} directly implies the second bound in the theorem.
\end{proof}

As an immediate consequence of the last theorem we obtain bounds on the density function of $\F_{q^\ell}$-codes with minimum distance bounded from below.
\begin{corollary} \label{cor:subbound}
Let $1 \le k \le ns$ and $1 \le d < +\infty$ be integers. We have
   \begin{align} \label{eq:upperBoundsublim}
    1-  \frac{\displaystyle (\textbf{v}_q^{D}(\F_{q^m}^n,d-1)-1)\dstirling{ns-1}{k-1}_{q^{\ell}} }{(q^{\ell}-1)\dstirling{ns}{k}_{q^{\ell}}}
    \le 
    \delta_q^D(\F_{q^m}^n,q^{\ell k},\ell,d) &\le 1-  \frac{\displaystyle (\textbf{v}_q^{D}(\F_{q^m}^n,d-1)-1)\dstirling{ns-1}{k-1}_{q^{\ell}} }{\bar\Theta(q^{\ell}-1)\dstirling{ns}{k}_{q^{\ell}}},
\end{align}
where $$\bar\Theta =  1+\dstirling{ns-1}{k-1}_{q^{\ell}}^{-1}\left(\displaystyle \frac{\textbf{v}_q^{D}(\F_{q^m}^n,d-1)-1}{q^{\ell}-1} -1 \right)\dstirling{ns-2}{k-2}_{q^{\ell}}.$$
\end{corollary}

\section{Asymptotic Results}
\label{sec:asy}
This section is devoted to general asymptotic results on the density function of codes in $\F_{q^m}^n$ endowed with a translation-invariant metric $D$. More precisely, we are interested in the following question: What is the probability that a uniformly random code in $(\F_{q^m}^n,D)$ of a given cardinality has minimum distance (at least) $d$, as either $q$, $n$, or in the (sub)linear case $s$ or $\ell$ tend to infinity, where $s := [\F_{q^m} \colon \F_{q^{\ell}}]$. Our results indicate that the answer to this question highly depends on the volume of balls in $(\F_{q^m}^n, D)$. We treat the nonlinear and (sub)linear case separately again.

\subsection{The Nonlinear Case}\label{sec:nonlinear-asym}

From Corollary~\ref{cor:nonlindensity} we obtain asymptotic bounds on the density (as $q \to +\infty$ or $n \to +\infty$) of (possibly) nonlinear codes with minimum distance bounded from below.

\begin{theorem} \label{thm:nonLinearLimq}
Let $2 \le n$ and let $1 \le d \le n$ be integers. Consider the sequence~$\smash{(\mathbb{F}_{q}^{n})_{q \in Q}}$ and a sequence of integers $\smash{(S_q)_{q \in Q}}$ with $2 \le S_q$ for all $q \in Q$. 
\begin{itemize}
    \item[(i)] We have
\begin{align*} 
  \max \left\{ \liminf_{q \to +\infty} \left( 1-\frac{\textbf{v}_q^{D}(\F_{q}^n,d-1)S_{q}^{2}}{2q^{n} }\right),0 \right\} \leq   \liminf_{q \to +\infty}\delta_q^D(\F_{q}^n,S_{q},0,d). 
\end{align*}
\item[(ii)]

If $\textbf{v}_q^{D}(\F_{q}^n,d-1) \in \Omega\left(q^{n}S_{q}^{-2}\right) \, \text{as} \, \, q \to +\infty, \, \text{then}$ 
\begin{align*} 
      \limsup_{q \to +\infty} \delta_q^D(\F_{q}^n,S_{q},0,d)\leq \limsup_{q \to +\infty}\left(\frac{1}{1+\frac{\textbf{v}_q^{D}(\F_q^n,d-1)S_q^2}{2q^{n}}}\right) <1.
\end{align*}

\end{itemize}
In particular,
\begin{align*}
  \lim_{q \to +\infty}\delta_q^D(\F_{q}^n,S_q,0,d) =   \begin{cases}
 1 \quad &\textnormal{ if $\textbf{v}_q^{D}(\F_{q}^n,d-1) \in o(q^{n}S_{q}^{-2})$ as $q \to +\infty$,}  \\
    0 \quad &\textnormal{ if $\textbf{v}_q^{D}(\F_{q}^n,d-1) \in \omega(q^{n}S_{q}^{-2})$ as $q \to +\infty$.} 
    \end{cases}
\end{align*}
\end{theorem}
\begin{proof}
The first statement of the theorem is an easy consequence of Corollary~\ref{cor:nonlindensity} and the fact that we always have $0 \le \delta_q^D(\F_q^n, S_q, 0, d)$. 
For the second statement, as in Theorem~\ref{thm:nonlinbound} we let
\begin{align*}
    \beta^0_q &= \frac{1}{2}q^n (\textbf{v}_q^{D}(\F_q^n,d-1)-1) - 2\textbf{v}_q^{D}(\F_q^n,d-1) +3, \\
    \beta^1_q &= 2\textbf{v}_q^{D}(\F_q^n,d-1)-4, \\
    \Theta_q &= 1 + \beta^1_q  \frac{S_q-2}{q^n-2} \, + \, \beta^0_q  \frac{(S_q-2)(S_q-3)}{\left(q^n-2\right)\left(q^n-3\right)}.
\end{align*}
It is easy to check that we have 
\begin{align*} 
\Theta_q \sim 1+ \frac{\textbf{v}_q^{D}(\F_q^n,d-1) S_q^2}{2q^n} \quad \textnormal{ as $q \to +\infty$.}
\end{align*}
Therefore,
\begin{align}
    \limsup_{q \to +\infty}\left(1- \frac{(\textbf{v}_q^{D}(\F_q^n,d-1)-1)S_q(S_q-1)}{2\Theta_q(q^n-1)}\right) &=  \limsup_{q \to +\infty}\left(1- \frac{\textbf{v}_q^{D}(\F_q^n,d-1)S_q^2}{2q^n+ \textbf{v}_q^{D}(\F_q^n,d-1)S_q^2}\right).    \label{eq:nonlin3}
\end{align}
In order to prove that the asymptotic upper bound in the theorem is smaller than 1, note that $\textbf{v}_q^{D}(\F_{q}^n,d-1) \in \Omega\left(q^{n}S_{q}^{-2}\right) \, \text{as} \, \, q \to +\infty$ means that
\begin{align*} 
    0 < \liminf_{q \to +\infty} \frac{\textbf{v}_q^{D}(\F_q^n,d-1)S_q^2}{q^n}.
\end{align*}
In particular, rewriting~\eqref{eq:nonlin3} gives
\begin{align*}
    \limsup_{q \to +\infty} \left(1-\frac{ \textbf{v}_q^{D}(\F_q^n,d-1)S_q^2}{2q^n+\textbf{v}_q^{D}(\F_q^n,d-1)S_q^2}\right) =  \limsup_{q \to +\infty} \left(\frac{1}{1+\frac{\textbf{v}_q^{D}(\F_q^n,d-1)S_q^2}{2q^{n}}}\right)< 1,
\end{align*}
concluding the proof of the theorem.
\end{proof}
We now establish the analogous result to Theorem~\ref{thm:nonLinearLimq} when $n$ tends to infinity. 
The proof of the following statement is an easy alteration of the proof of Theorem~\ref{thm:nonLinearLimq} and hence we omit it.

\begin{theorem} \label{thm:nonLinearLimn}
Let $q \in Q$ be a prime power and $d$ be an integer such that $1 \le d < +\infty$. Consider the sequence $(\mathbb{F}_{q}^{n})_{n \geq 1}$ and a sequence of integers $(S_n)_{n \geq 1}$ with $2 \le S_n$ for all $1 \le n$. 
\begin{itemize}
    \item[(i)] We have
\begin{align*} 
  \max \left\{ \liminf_{n \to +\infty} \left( 1-\frac{\textbf{v}_q^{D}(\F_{q}^n,d-1)S_{n}^{2}}{2q^{n} }\right),0 \right\} \leq   \liminf_{n \to +\infty}\delta_q^D(\F_{q}^n,S_{n},0,d). 
\end{align*}
\item[(ii)]
If $\textbf{v}_q^{D}(\F_{q}^n,d-1) \in \Omega\left(q^{n}S_{n}^{-2}\right) \, \text{as} \, \, n \to +\infty, \, \text{then}$ 
\begin{align*} 
      \limsup_{n\to +\infty} \delta_q^D(\F_{q}^n,S_{n},0,d)\leq \limsup_{n \to +\infty}\left(\frac{1}{1+\frac{\textbf{v}_q^{D}(\F_q^n,d-1)S_q^2}{2q^{n}}}\right) <1.
\end{align*}
\end{itemize}
In particular,
\begin{align*}
  \lim_{n \to +\infty}\delta_q^D(\F_{q}^n,S_n,0,d) =   \begin{cases}
 1 \quad &\textnormal{ if $\textbf{v}_q^{D}(\F_{q}^n,d-1) \in o(q^{n}S_{n}^{-2})$ as $n \to +\infty$,}  \\
    0 \quad &\textnormal{ if $\textbf{v}_q^{D}(\F_{q}^n,d-1) \in \omega(q^{n}S_{n}^{-2})$ as $n \to +\infty$.} 
    \end{cases}
\end{align*}
\end{theorem}

As we applied the bounds given in Corollary~\ref{cor:nonlindensity} in a very general way, without having any knowledge about $\smash{\textbf{v}_{q}^{D}(\F_{q}^n,r)}$, the requirements on the asymptotic behavior on the cardinalities in Theorem~\ref{thm:nonLinearLimq} and Theorem~\ref{thm:nonLinearLimn} for certain codes to be dense or sparse are the same, even though we send different parameters to infinity. This is something that we will also observe when looking at (sub)linear codes in the next subsection.

\begin{remark} \label{rem:gilbvarsh}
In~\cite{gilbert1952comparison,varshamov1957estimate} Gilbert and Varshamov gave a lower bound that shows the existence of codes of sufficiently large cardinality and minimum distance. More precisely, it says that there exists a code $\mC\subseteq \F_{q}^n$ of minimum distance $d$ with
$$\frac{q^{n}}{\textbf{v}_q^{D}(\mathbb{F}_{q}^n,d-1)} \le |\mC|.$$
It is an immediate consequence of Theorem~\ref{thm:nonLinearLimq} and Theorem~\ref{thm:nonLinearLimn} that, while such codes exist, the probability that a uniformly random (possibly nonlinear) code, whose cardinality is close to the Gilbert-Varshamov bound, has minimum distance at least $d$, goes to 0 both as $q \to +\infty$ and $n \to +\infty$.
\end{remark}

\subsection{The (Sub)Linear Case} \label{sec:(Sub)linear_asym}
From Corollary~\ref{cor:subbound} we obtain results for the asymptotic density of (sub)linear codes with minimum distance bounded from below as one of the parameters $q,n,s$ or $\ell$ tends to infinity and the other three parameters are treated as constants.
We will repeatedly use some asymptotic estimates of the $q$-binomial coefficient. One of these estimates involves the quantity
\begin{align*}
    \pi(q) := \prod_{i=1}^{\infty} \left( \frac{q^{i}}{q^{i}-1} \right).
\end{align*}
Note that the infinite product $\pi(q)$ is closely linked to the Euler function $\phi$; see e.g.~\cite[Section 14]{apostol2013introduction}. The Euler function $\phi:(-1,1) \to \R$ is defined as
\begin{equation*} 
\phi  :  x \mapsto \prod_{i=1}^{\infty}(1-x^i).
\end{equation*}
We have $\pi(q)=1/\phi(1/q)$ for all $q \in Q$.
We will also need the following asymptotic estimates.
\begin{lemma} \label{lem:asymptotics}
Let $0 \le b \le a$ be integers and let $k  :  \mathbb{N} \rightarrow \mathbb{N}$ and $\tilde{k}  :  \mathbb{N} \rightarrow \mathbb{N}$ be linear functions such that $\tilde{k}(n) \le k(n)$ for all $n \in \mathbb{N}$. We have 
\begin{enumerate}[label=(\alph*)]\setlength\itemsep{0cm}
\item
\label{lem:asymptoticss}
   $ {\dstirling{k(n)}{\tilde{k}(n)}_{q^{\ell}}} \sim q^{\ell \tilde{k}(n)(k(n)-\tilde{k}(n))}\pi(q^{\ell})$ \quad \textnormal{as $n \to +\infty$.}

\item \label{lem:asymptoticsq}
  $  {\dstirling{a}{b}_{q^{\ell}}} \sim q^{\ell b(a-b)}$ \quad \textnormal{both as $q \to +\infty$, (respectively $\ell \to +\infty$.)}
\end{enumerate}
\end{lemma}

Similar asymptotic estimates as in Lemma~\ref{lem:asymptotics} have been proved in~\cite[Section 6]{gruica2022common}, and thus we omit the proofs.

\begin{theorem} \label{thm:sublimq}
    Let $1 \le n,\ell,s$ and $1 \le d < +\infty$ be integers. We fix $1 \le k \le ns$ and consider the sequence $\smash{(\F_{q^{m}}^n)_{q \in Q}}$.
\begin{itemize}
    \item[(i)] We have
\begin{align*} 
  \max\left\{ \liminf_{q \to +\infty} \left( 1-\frac{\textbf{v}_q^{D}(\F_{q^{m}}^n,d-1)}{q^{\ell(ns+1-k)}} \right),0 \right\} \leq   \liminf_{q \to +\infty}\delta_q^D(\F_{q^{m}}^n,q^{\ell k},\ell,d). 
\end{align*}
\item[(ii)]
If $\textbf{v}_q^{D}(\F_{q^{m}}^n,d-1) \in \Omega\left(q^{\ell(ns+1-k)}\right) \, \text{as} \, \, q \to +\infty, \, \text{then}$ 
\begin{align} \label{eq:asymupperBound_q}
      \limsup_{q \to +\infty} \delta_q^D(\F_{q^{m}}^n,q^{\ell k},\ell,d)\leq \limsup_{q \to +\infty} \left( \frac{1}{1 + {\textbf{v}_q^{D}(\F_{q^{m}}^n,d-1)}{q^{-\ell(ns+1-k)}}}\right) <1.
\end{align}
\end{itemize}
In particular,
\begin{align*}
  \lim_{q \to +\infty}\delta_q^D(\F_{q^{m}}^n,q^{\ell k},\ell,d)  = \begin{cases}
 1 \quad &\textnormal{ if $\textbf{v}_q^{D}(\F_{q^{m}}^n,d-1) \in o(q^{\ell(ns+1-k)})$ as $q \to +\infty$,}  \\
    0 \quad &\textnormal{ if $\textbf{v}_q^{D}(\F_{q^{m}}^n,d-1) \in \omega(q^{\ell(ns+1-k)})$ as $q \to +\infty$.} 
    \end{cases}
\end{align*}
\end{theorem} 
\begin{proof}
From Lemma~\ref{lem:asymptotics} we obtain
\begin{align*} \label{eq:asymlower_q}
    \frac{(\textbf{v}_q^{D}(\F_{q^{m}}^n,d-1)-1)\dstirling{ns-1}{k-1}_{q^{\ell}}}{(q^{\ell}-1)\dstirling{ns}{k}_{q^{\ell}}} \sim \frac{\textbf{v}_q^{D}(\F_{q^{m}}^n,d-1)}{q^{\ell(ns+1-k)}} \quad \text{as} \, \, q \to +\infty.
\end{align*}
Now (i) is an easy consequence of Corollary ~\ref{cor:subbound} and the fact that $0 \le \delta_q^D(\F_{q^{m}}^n,q^{\ell k},\ell,d)$.

For the second statement we consider the upper bound of Corollary~\ref{cor:subbound}. Together with Lemma~\ref{lem:asymptotics} we have 
\begin{align*}
    \left( \displaystyle \frac{\textbf{v}_q^{D}(\F_{q^{m}}^n,d-1)-1}{q^{\ell}-1} \right) \dstirling{ns-1}{k-1}_{q^{\ell}}^{2} \sim \textbf{v}_q^{D}(\F_{q^{m}}^n,d-1)q^{2\ell(k-1)(ns-k)-\ell}
\end{align*}
and 
\begin{align*}
\dstirling{ns}{k}_{q^{\ell}}\left( \dstirling{ns-1}{k-1}_{q^{\ell}}+\left(\displaystyle \frac{\textbf{v}_q^{D}(\F_{q^{m}}^n,d-1)-1}{q^{\ell}-1} -1 \right)\dstirling{ns-2}{k-2}_{q^{\ell}}\right) \\
\sim 
q^{\ell(2k-2)(ns-k)-\ell}(q^{\ell(ns-k+1)} + \textbf{v}_q^{D}(\F_{q^{m}}^n,d-1) )
\end{align*}
 as $q \to +\infty$. Hence we obtain
\begin{align*}
 \frac{\left( \displaystyle \frac{\textbf{v}_q^{D}(\F_{q^{m}}^n,d-1)-1}{q^{\ell}-1} \right) \dstirling{ns-1}{k-1}_{q^{\ell}}^{2}}{\dstirling{ns}{k}_{q^{\ell}}\left( \dstirling{ns-1}{k-1}_{q^{\ell}}+\left( \displaystyle\frac{\textbf{v}_q^{D}(\F_{q^{m}}^n,d-1)-1}{q^{\ell}-1} -1 \right)\dstirling{ns-2}{k-2}_{q^{\ell}}\right)} 
  \sim   \frac{\textbf{v}_q^{D}(\F_{q^{m}}^n,d-1)}{q^{\ell(ns-k+1)} + \textbf{v}_q^{D}(\F_{q^{m}}^n,d-1)},
\end{align*}
 as $q \to +\infty$.
Now, taking the limit superior as $q \to +\infty$ in the bound of ~\eqref{eq:upperBoundsublim} reads
\begin{align*} 
 \limsup_{q \to +\infty}\delta_q^D(\F_{q^m}^n,q^{\ell k},\ell,d) &\le  \limsup_{q \to +\infty} \left(1- \frac{\textbf{v}_q^{D}(\F_{q^{m}}^n,d-1)}{q^{\ell(ns-k+1)} + \textbf{v}_q^{D}(\F_{q^{m}}^n,d-1)} \right).
\end{align*}
Since $q^{\ell k} \in \Omega\left(\frac{q^{\ell(ns+1)}}{\textbf{v}_q^{D}(\F_{q^{m}}^n,d-1)}\right) \, \text{as} \, \, q \to +\infty$, we have
\begin{align*}
    \limsup_{q \to +\infty} \left(1- \frac{\textbf{v}_q^{D}(\F_{q^{m}}^n,d-1)}{q^{\ell(ns-k+1)} + \textbf{v}_q^{D}(\F_{q^{m}}^n,d-1)} \right) = \limsup_{q\to +\infty}\frac{1}{1+\textbf{v}_q^{D}(\F_{q^{m}}^n,d-1)q^{-\ell(ns+1-k)}} <1 ,
\end{align*}
yielding the statement in~\eqref{eq:asymupperBound_q}.
\end{proof} 
Note that the asymptotic estimates for the $q$-binomial coefficient given in Lemma~\ref{lem:asymptotics} are the same as $q \to +\infty$, respectively as $\ell \to +\infty$.
Hence we can give an analogous result to Theorem~\ref{thm:sublimq} but with $\ell$ tending to infinity.  

Further note that the parameters of the $q$-binomial coefficient involved in the upper and lower bound of $\delta_q^D(\F_{q^{m}}^n,q^{\ell k},\ell,d)$ (see Corollary~\ref{cor:subbound}) can be written as a linear function in either $n$ or in $s$ such that the assumptions of
Lemma~\ref{lem:asymptotics} are fulfilled. This gives us then two asymptotic results regarding $n \to +\infty$ (and $s \to +\infty$). Since the proofs of Theorem~\ref{thm:sublimell}, Theorem~\ref{thm:sublimn} and Theorem~\ref{thm:sublims} follow the same arguments as the proof of Theorem~\ref{thm:sublimq} we omit them.

\begin{theorem}\label{thm:sublimell}
    Let $1 \le n,s$ and $1 \le d < +\infty$ be integers and let $q \in Q$. We fix $1 \le k \le ns$ and consider the sequence $(\F_{q^{\ell s}}^n)_{\ell \geq 1}$.
\begin{itemize}
    \item[(i)] We have
\begin{align*} 
  \max\left\{ \liminf_{\ell \to +\infty} \left( 1-\frac{\textbf{v}_q^{D}(\F_{q^{\ell s}}^n,d-1)}{q^{\ell(ns+1-k)}} \right),0 \right\} \leq   \liminf_{\ell \to +\infty}\delta_q^D(\F_{q^{\ell s}}^n,q^{\ell k},\ell,d). 
\end{align*}
\item[(ii)]
If $\textbf{v}_q^{D}(\F_{q^{\ell s}}^n,d-1) \in \Omega\left(q^{\ell(ns+1-k)}\right) \, \text{as} \, \, \ell \to +\infty, \, \text{then}$ 
\begin{align*} 
      \limsup_{\ell \to +\infty} \delta_q^D(\F_{q^{\ell s}}^n,q^{\ell k},\ell,d)\leq \limsup_{\ell \to +\infty}\left(\frac{1}{1 + \textbf{v}_q^{D}(\F_{q^{\ell s}}^n,d-1)q^{-\ell(ns+1-k)}}\right) <1.
\end{align*}
\end{itemize}
In particular,
\begin{align*}
  \lim_{\ell \to +\infty}\delta_q^D(\F_{q^{\ell s}}^n,q^{\ell k},\ell,d) =   \begin{cases}
 1 \quad &\textnormal{ if $\textbf{v}_q^{D}(\F_{q^{\ell s}}^n,d-1) \in o(q^{\ell(ns+1-k)})$ as $\ell \to +\infty$,}  \\
    0 \quad &\textnormal{ if $\textbf{v}_q^{D}(\F_{q^{\ell s}}^n,d-1) \in \omega(q^{\ell(ns+1-k)})$ as $\ell \to +\infty$.} 
    \end{cases}
\end{align*}
\end{theorem} 
In the next theorem we consider the asymptotic density, where $n$ tends to infinity and the other parameters are treated as constants. 

\begin{theorem}\label{thm:sublimn}
    Let $q \in Q, 1 \le s, \ell$ and $1 \le d < +\infty $ be integers. Define a linear function $k : \mathbb{N} \rightarrow \mathbb{N}$ such that $(k(n))_{n \ge 1}$ describes a sequence of integers where $1 \le k(n) \le ns$ for all $1 \le n$ and consider the sequence $(\F_{q^{m}}^n)_{n \geq 1}$  . 
\begin{itemize}
    \item[(i)] We have
\begin{align*} 
  \max\left\{ \liminf_{n \to +\infty} \left( 1-\frac{\textbf{v}_q^{D}(\F_{q^{m}}^n,d-1)}{q^{\ell(ns+1-k)}} \right),0 \right\} \leq   \liminf_{n \to +\infty}\delta_q^D(\F_{q^{m}}^n,q^{\ell k},\ell,d). 
\end{align*}
\item[(ii)]
If $\textbf{v}_q^{D}(\F_{q^{m}}^n,d-1) \in \Omega\left(q^{\ell(ns+1-k)}\right) \, \text{as} \, \, n \to +\infty, \, \text{then}$ 
\begin{align*} 
      \limsup_{n \to +\infty} \delta_q^D(\F_{q^{m}}^n,q^{\ell k},\ell,d)\leq \limsup_{n \to +\infty}\left(\frac{1}{1 + \textbf{v}_q^{D}(\F_{q^{m}}^n,d-1)q^{-\ell(ns+1-k)}}\right) <1.
\end{align*}
\end{itemize}
In particular,
\begin{align*}
  \lim_{n \to +\infty}\delta_q^D(\F_{q^{m}}^n,q^{\ell k},\ell,d) =   \begin{cases}
 1 \quad &\textnormal{ if $\textbf{v}_q^{D}(\F_{q^{m}}^n,d-1) \in o(q^{\ell(ns+1-k)})$ as $n \to +\infty$,}  \\
    0 \quad &\textnormal{ if $\textbf{v}_q^{D}(\F_{q^{m}}^n,d-1) \in \omega(q^{\ell(ns+1-k)})$ as $n \to +\infty$.} 
    \end{cases}
\end{align*}
\end{theorem} 
In the last theorem concerning the asymptotic density of $\F_{q^{\ell}}$-linear codes we let $s$, go to infinity. 
\begin{theorem}\label{thm:sublims}
    Let $q \in Q$, $1 \le n, \ell$ and $1 \le d < +\infty$ be integers. Define a linear function $k : \mathbb{N} \rightarrow \mathbb{N}$ such that $(k(s))_{s \ge 1}$ describes a sequence of integers with $1 \le k(s) \le ns$ for all $1 \le s$ and consider the sequence $(\F_{q^{\ell s}}^n)_{s \geq 1}$.
\begin{itemize}
    \item[(i)] We have
\begin{align*} 
  \max\left\{ \liminf_{s \to +\infty} \left( 1-\frac{\textbf{v}_q^{D}(\F_{q^{\ell s}}^n,d-1)}{q^{\ell(ns+1-k)}} \right),0 \right\} \leq   \liminf_{s \to +\infty}\delta_q^D(\F_{q^{\ell s}}^n,q^{\ell k},\ell,d). 
\end{align*}
\item[(ii)]
If $\textbf{v}_q^{D}(\F_{q^{\ell s}}^n,d-1) \in \Omega\left(q^{\ell(ns+1-k)}\right) \, \text{as} \, \, s \to +\infty, \, \text{then}$ 
\begin{align*}
      \limsup_{s \to +\infty} \delta_q^D(\F_{q^{\ell s}}^n,q^{\ell k},\ell,d)\leq \limsup_{s \to +\infty}\left(\frac{1}{1 + \textbf{v}_q^{D}(\F_{q^{\ell s}}^n,d-1)q^{-\ell(ns+1-k)}}\right) <1.
\end{align*}
\end{itemize}
In particular,
\begin{align*}
  \lim_{s \to +\infty}\delta_q^D(\F_{q^{\ell s}}^n,q^{\ell k},\ell,d) =   \begin{cases}
 1 \quad &\textnormal{ if $\textbf{v}_q^{D}(\F_{q^{\ell s}}^n,d-1) \in o(q^{\ell(ns+1-k)})$ as $s \to +\infty$,}  \\
    0 \quad &\textnormal{ if $\textbf{v}_q^{D}(\F_{q^{\ell s}}^n,d-1) \in \omega(q^{\ell(ns+1-k)})$ as $s \to +\infty$.} 
    \end{cases}
\end{align*}
\end{theorem} 

\begin{remark} \label{rem:gilbvarsh2}
Since the term ${q^{\ell ns}}/{q^{\ell(ns+1)}}$ tends to zero both as $q \to +\infty$, respectively $\ell \to +\infty$ Theorem~\ref{thm:sublimq} and Theorem~\ref{thm:sublimell} imply that a uniformly random $\F_{q^{\ell}}$-linear code in $\F_{q^m}^n$ of minimum distance at least $d$, for some integer $1 \le d < +\infty$, will attain the Gilbert-Varshamov bound with high probability for large $q$ and also for large linearity degree $\ell$.
However, when either $n$ or $s$ tend to infinity (and the other parameters are treated as constants), then it can easily be seen that the probability that $\F_{q^\ell}$-linear codes attain the Gilbert-Varshamov bound is upper bounded by $q^{\ell}/(q^{\ell}+1)$.
\end{remark}

\begin{theorem}
Let $(\F_{q^m}^n,D)$ be a metric space where $D$ is a translation-invariant metric on $\F_{q^m}^n$ and let $\ell$ be a divisor of $m$. The probability that a uniformly random $\F_{q^\ell}$-linear code in $(\F_{q^m}^n,D)$ satisfies the Gilbert-Varshamov bound is (infinitesimally close to) 1 if we let $q$ go to infinity.
\end{theorem}

\section{Application I: Codes in the Hamming Metric} \label{sec:Hamming}

In this section we determine the asymptotic density of codes in $\F_{q^{m}}^{n}$, equipped with the Hamming metric. In particular, we study the probability that a uniformly random code in $\F_{q^{m}}^{n}$ of a given cardinality achieves the Singleton bound with equality as one of the four parameters $q,n,s$ or $\ell$ tends to infinity (where $m=s\ell$) and the other three parameters are treated as constants.
 We start by introducing the required preliminaries.

\begin{definition}
Let $x \in \F_{q^m}^n$. The \emph{Hamming weight} of $x$ is $\wH(x)$ where $\wH$ is the function defined as
    \begin{equation*}
       \wH  :  \F_{q^m}^n \longrightarrow \mathbb{N}, \hspace{0.5em} x \mapsto |\{ i \in [n] :   x_{i}  \neq 0 \}|.
    \end{equation*}
Let $x,y \in \F_{q^m}^n$, then the \emph{Hamming distance} between $x$ and $y$ is $\dH(x,y) := \wH(x-y)$.
\end{definition}
Throughout this section, we are working in the metric space $(\F_{q^m}^n,\dH)$. As in Section~\ref{sec:prelim}, we denote nonlinear codes of cardinality $S$ in $\F_{q^m}^n$ and of Hamming distance at least $d$ as $[\F_{q^m}^n,S,0,d]^{H}$-codes. Similarly, we call $\F_{q^\ell}$-linear codes of dimension $k$ and minimum Hamming distance at least $d$ $[\F_{q^m}^n,q^{\ell k},\ell,d]^{H}$-codes.

It is a desirable property for a code to have large cardinality and large minimum distance at the same time. The trade-off between the cardinality and the minimum distance of a Hamming-metric code is captured by the following famous result. 

\begin{theorem}[Singleton bound; \text{see \cite[Theorem 5.4]{delsarte1978bilinear}}] \label{thm:singletonlikeHamming}
Let $\mC \subseteq \mathbb{F}_{q^{m}}^{n}$ be a $[\F_{q^m}^n,q^{\ell k},\ell,d]^{H}$-code. We have $q^{\ell k} \le q^{m(n-d+1)}.$
\end{theorem}

A code in $(\F_{q^m}^n, \dH)$ is called \emph{MDS} (\emph{maximum distance separable}) \emph{code} if its cardinality meets the Singleton bound with equality.

To estimate the asymptotic density of  MDS codes we need the size of the Hamming-metric ball and its estimates. It is well-known and easy to see that the volume of the Hamming-metric ball of radius $r$ is
\begin{equation} \label{eq:sizeHammingVolume}
     \textbf{v}_{q}^{\textnormal{H}}(\F_{q^m}^n,r):= \sum_{i=0}^r\binom{n}{i}(q^{m}-1)^{i},
\end{equation}
for any $0 \le r < \infty$.
The following lemma states the asymptotic estimates of the Hamming-metric ball as $q \to +\infty$, as $n\to +\infty$ and as $m\to +\infty$. All three statements can easily be derived from Equation~\eqref{eq:sizeHammingVolume}.

\begin{lemma} \label{lem:hamballasy}
The following estimates hold.
\begin{itemize}
    \item[(i)] Let $0 \le r \le n$. We have $\textbf{v}_{q}^{\textnormal{H}}(\F_{q^m}^n,r) \sim \binom{n}{r}q^{rm}$ as $q \to +\infty$.
    \item[(ii)] Let $0 \le r \le n$. We have $\textbf{v}_{q}^{\textnormal{H}}(\F_{q^m}^n,r) \sim \binom{n}{r}q^{rm}$ as $m \to +\infty$.
    \item[(iii)] Let $0 \le r < +\infty$. We have $\textbf{v}_{q}^{\textnormal{H}}(\F_{q^m}^n,r) \sim \binom{n}{r}(q^{m}-1)^r$ as $n \to +\infty$.
\end{itemize}

\end{lemma}
\subsection{Nonlinear MDS Codes}
In this short subsection we determine the density of nonlinear MDS codes in $\F_q^n$ (i.e., we set $m=1$, as before) as either $q$ or $n$ tends to infinity. We summarize the obtained results in the following theorem. 

\begin{theorem} \label{thm:nonlinearnHamming}
Let $2 \leq d$ be an integer.
\begin{itemize}
\item[(i)]
Let $2 \le n$ be an integer. Then we have
$$\lim_{q \to +\infty} \delta_q^{H}(\F_{q}^n,q^{n-d+1},0,d)= 0.$$
\item[(ii)]
We have
$$\lim_{n \to +\infty} \delta_q^{H}(\F_{q}^n,q^{n-d+1},0,d)= 0.$$
\end{itemize}
In particular, nonlinear MDS codes in $\F_q^n$ are sparse both as $q \to +\infty$, respectively as $n \to +\infty$.
\begin{proof}
From the estimates given in Lemma~\ref{lem:hamballasy} we get
\begin{align*}
\lim_{q \to +\infty}\frac{q^{n}}{\textbf{v}_{q}^{\textnormal{H}}(\F_{q}^{n},d-1)q^{2(n-d+1)}} = \lim_{q \to +\infty}\frac{1}{\binom{n}{d-1}q^{n-d+1}} =0
\end{align*}
which, by Theorem~\ref{thm:nonLinearLimq}, gives the first result.

Similarly, we have
\begin{align*}
\lim_{n \to +\infty}\frac{q^{n}}{\textbf{v}_{q}^{\textnormal{H}}(\F_{q}^{n},d-1)q^{2(n-d+1)}} = \lim_{n \to +\infty}\frac{1}{\binom{n}{d-1}q^{n}}\left( \frac{q^{2}}{q-1}\right)^{d-1} =0
\end{align*}
which, together with Theorem~\ref{thm:nonLinearLimn}, implies the second result.
\end{proof}
\end{theorem}

\subsection{(Sub)linear MDS Codes}
Since we are now in the (sub)linear case, we fix a divisor $\ell$ of $m$ such that $\F_{q^{\ell}}$ is a subfield of $\F_{q^{m}}$ and let $s:= [\F_{q^{m}} \colon \F_{q^{\ell}}]$. We then apply the results of Section~\ref{sec:(Sub)linear_asym} to derive density results of  $[\F_{q^m}^n,q^{m(n-d+1)},\ell,d]^{H}$-codes when one of the four parameters $q, \ell, n$ or $s$ tends to infinity.

First we consider the asymptotic density of $\F_{q^{\ell}}$-linear MDS codes as the field size $q$ tends to infinity.
\begin{theorem} \label{thm:sublinHammingq}
Let $1 \leq n, \ell, s,d$ be integers such that $1 \leq d \leq n$.
Then we have
\begin{align*}
\lim_{q \to +\infty} \delta_q^{H}(\F_{q^{m}}^n,q^{m(n-d+1)},\ell,d)= 1,
\end{align*}
i.e., MDS codes in $\F_{q^m}^n$ are dense as $q \to +\infty.$
\end{theorem}
\begin{proof}
From the estimate of the volume given in Lemma~\ref{lem:hamballasy} we have
\begin{align*} 
\lim_{q \to +\infty} \frac{q^{m(n-d+1)} \, \textbf{v}_{q}^{\textnormal{H}}(\F_{q^{m}}^n,d-1)}{q^{\ell(ns+1)}} = \lim_{q \to +\infty} \frac{\binom{n}{d-1}}{q^{\ell}}=0
\end{align*}
which, by Theorem~\ref{thm:sublimq}, gives the statement of the theorem.
\end{proof}
Analogously, by the asymptotic formulas given in Theorem~\ref{thm:sublimell}, Theorem~\ref{thm:sublinHammingq} reads as follows for increasing field extension degree $\ell$.
\begin{theorem} \label{thm:sublinHammingell}
Let $q \in Q$ and $1 \leq n, s,d$ be integers such that $1 \leq d \leq n$.
Then
\begin{align*}
\lim_{\ell \to +\infty} \delta_q^{H}(\F_{q^{\ell s}}^n,q^{\ell s(n-d+1)},\ell,d)= 1.
\end{align*}
\end{theorem}
If we let the code length $n$ tend to infinity, we obtain the following sparsity result. We state this result for completeness, even though it is well known that MDS codes do not exist at all for large $n$; see Remark \ref{rem:MDSconj}.

\begin{theorem} \label{thm:sublinHammingn}
Let $q \in Q$ and $1 \leq \ell, s$, $2 \leq d < \infty$ be integers. We have
\begin{align*}
\lim_{n \to +\infty} \delta_q^{H}(\F_{q^{m}}^n,q^{m(n-d+1)},\ell,d)= 0.
\end{align*}
\begin{proof}
By Lemma~\ref{lem:hamballasy} we have
\begin{align*}
  \textbf{v}_{q}^{\textnormal{H}}(\F_{q^{m}}^n,d-1) \sim \binom{n}{d-1}(q^{m}-1)^{d-1} \quad \textnormal{as} \, n \to + \infty.  
\end{align*}
One easily checks that $\textbf{v}_{q}^{\textnormal{H}}(\F_{q^{m}}^n,d-1) \in \omega(q^{\ell (s(d-1)+ 1)}).$ This combined with Theorem~\ref{thm:sublimn} gives the statement of the theorem.
\end{proof}
\end{theorem}
The last parameter we consider is the degree $s$ of the field extension $\F_{q^{m}}/ \F_{q^{\ell}}$. For this parameter we will only upper bound the asymptotic density $\lim_{s \to +\infty} \delta_q^{H}(\F_{q^{\ell s}}^n,q^{\ell s(n-d+1)},\ell,d)$.

\begin{theorem} 
Let $q \in Q$ and $2 \leq n$, $1 \leq \ell$, $2 \leq d \leq n$ be integers.
Then
\begin{align*}
\limsup_{s \to +\infty} \delta_q^{H}(\F_{q^{\ell s}}^n,q^{\ell s(n-d+1)},\ell,d)\leq \frac{1}{1+ \binom{n}{d-1}q^{-\ell}} <1.
\end{align*}
\end{theorem}
\begin{proof}
From the estimate of $\textbf{v}_{q}^{\textnormal{H}}(\F_{q^{\ell s}}^n,d-1)$ given in Lemma~\ref{lem:hamballasy} we get
\begin{align} \label{eq:sublinHammings}
\lim_{s \to +\infty} \frac{q^{\ell s(n-d+1)} \cdot \textbf{v}_{q}^{\textnormal{H}}(\F_{q^{\ell s}}^n,d-1)}{q^{\ell(ns+1)}} = \frac{\binom{n}{d-1}}{q^{\ell}}.
\end{align}
Hence $\textbf{v}_{q}^{\textnormal{H}}(\F_{q^{\ell s}}^n,d-1) \in \Omega(q^{\ell (s(d-1)+ 1)})$ as $s\to + \infty$ and applying Theorem~\ref{thm:sublims} gives
\begin{align} \label{eq:upperBoundHamming}
    \limsup_{s \to +\infty} \delta_q^{H}(\F_{q^{\ell s}}^n,q^{\ell s(n-d+1)},\ell,d) \leq \limsup_{s \to +\infty}\left( \frac{1}{1+\textbf{v}_{q}^{\textnormal{H}}(\F_{q^{\ell s}}^n,d-1)q^{-\ell(s(d-1)+1)}}\right).
\end{align}
Now the desired upper bound follows from \eqref{eq:sublinHammings} and \eqref{eq:upperBoundHamming}. 
\end{proof}

 \begin{remark}\label{rem:MDSconj}
Note that fixing the linearity degree of the code and studying the asymptotic density of MDS codes for the parameters $q,n, \ell$ and $s$ leads to different results. For growing field size, i.e.,  $q\to + \infty$ (or $\ell \to + \infty$), $[\F_{q^{m}}^n,q^{m(n-d+1)},\ell,d]^{H}$-codes are dense. On the contrary, if we let $n$ grow, then MDS codes are sparse. This is in line with the MDS conjecture, which implies that no non-trivial $[\F_{q^m}^{n}, q^{m(n-d+1)}, \ell, d]$-codes exist if $n>q^m+2$.
Considering the asymptotic density for $s \to +\infty$, one can see that $[\F_{q^{m}}^n,q^{m(n-d+1)},\ell,d]^{H}$-codes are not dense, but the question whether MDS codes are sparse or not for large $s$ remains open.
\end{remark}

\section{Application II: Codes in the Rank Metric} \label{sec:Rank}

In this section we apply the results of Sections~\ref{sec:bounds} and~\ref{sec:asy} to codes in the rank metric. We start by quickly recalling the required preliminaries on rank-metric codes.

\begin{definition}
Let $x \in \F_{q^m}^n$. The \emph{rank weight} of $x$ is defined as the dimension of the $\mathbb{F}_{q}$-span of its entries. More formally, for $x \in \F_{q^m}^n$ we define $\wrk(x)$ as 
\begin{equation*}
        \wrk  :  \F_{q^m}^n \longrightarrow \mathbb{N}, \hspace{0.5em} x \mapsto  \dim_{\mathbb{F}_{q}}\langle x_{1}, \dots , x_{n} \rangle.
\end{equation*}
The \emph{rank distance} between $x,y \in \F_{q^m}^n$ is then defined as $\drk(x,y):=\wrk(x-y)$.
\end{definition}

One can show that $\drk$ is a translation-invariant metric on $\F_{q^m}^n$ and throughout this section we always work with the metric space $(\F_{q^m}^n, \drk)$. At times we will take advantage of the following observation.

\begin{remark} \label{rem:matrixiso}
 When we equip the matrix space $\F_q^{m\times n}$ with the metric $\tilde{D}^\textnormal{rk}(X,Y):=\rk(X-Y)$, for all $X,Y \in \F_q^{m\times n}$, then  
a suitable vector space isomorphism from $(\F_{q^m}^n, D^\rk)$ to $(\F_q^{m\times n}, \tilde{D}^\textnormal{rk})$ is also an isometry, see e.g.~\cite{gorla2018codes}. 
 However, the linearity degree of a code is generally not preserved. In particular, the image of an $\F_{q^\ell}$-linear code in $\F_{q^m}^{n}$ is ``only'' $\F_q$-linear in $\F_q^{m\times n}$, for any $1\leq \ell \leq m$. Nonetheless, for $\ell\in \{0,1\}$, the set of $[\F_{q^m}^n, S, \ell, d]^\rk$-codes is in one-to-one-correspondence with the set of $[\F_{q^n}^m, S, \ell, d]^\rk$-codes.\footnote{If $\ell$ is also a divisor of $n$, then there exist $\F_{q^\ell}$-isomorphisms from $\F_{q^m}^n$ to $\F_{q^\ell}^{s\times n}\cong \F_{q^\ell}^{sn}$, and from $\F_{q^\ell}^{sn}\cong \F_{q^\ell}^{m\times (n/\ell)}$ to  $ \F_{q^n}^{m}$. Defining the $\F_q$-rank metric in a suitable manner on $\F_{q^\ell}^{sn}$, we get that these isomorphisms are isometries again. Then the set of  $[\F_{q^m}^n, S, \ell, d]^\rk$-codes is in one-to-one-correspondence with the set of $[\F_{q^n}^m, S, \ell, d]^\rk$-codes, also for larger $\ell$.}
\end{remark}

As usual, we denote a nonlinear code in $(\F_{q^m}^n, \drk)$ of cardinality $S$ and minimum distance at least $d$ as a $[\F_{q^m}^n,S,0,d]^\textnormal{rk}$-code. If the code is also $\F_{q^\ell}$-linear and then we say it is a $[\F_{q^m}^n,S,\ell,d]^\textnormal{rk}$-code.

A rank-metric code cannot both have large dimension and minimum distance. The following result by Delsarte shows the relation between these two quantities.

\begin{theorem}[Singleton-like bound; \text{see \cite[Theorem 5.4]{delsarte1978bilinear}}] \label{thm:singletonlike}
Let $\mC \subseteq \F_{q^m}^n$ be a rank-metric code. We have $|\mC| \le q^{\max\{n,m\}(\min\{n,m\}-\drk(\mC)+1)}.$
\end{theorem}

We call a rank-metric code meeting the bound in Theorem~\ref{thm:singletonlike} with equality an \emph{MRD} (\emph{maximum rank distance}) code.

\begin{remark} \label{rem:quasiMRD}
Let $\ell$ be a divisor of $m$ and let $s=[\F_{q^{m}}\colon \F_{q^{\ell}}]$. The dimension over $\F_q$ of any $\F_{q^\ell}$-linear rank-metric code $\mC \subseteq \F_{q^{\ell s}}^n$ has to be divisible by $\ell$. In particular, if the dimension of $\mC$ is $k$ over $\F_{q^\ell}$ where $0 \le k \le ns$ and $d=\drk(\mC)$, then from the bound in Theorem \ref{thm:singletonlike}, it follows
\begin{align} \label{eq:largestk}
    \ell k \le \max\{n,\ell s\}(\min\{n,\ell s\}-d+1).
\end{align}
If $\ell s \le n$, then the largest integer $k$ that satisfies \eqref{eq:largestk} is 
\begin{align}  \label{eq:largestk2}
    k^{*}=\left\lfloor \frac{n(\ell s -d+1)}{\ell}\right\rfloor.
\end{align}
Clearly, codes attaining the largest possible dimension $k^*$ in~\eqref{eq:largestk2} are not necessarily MRD (it depends on whether or not the fraction on the RHS of \eqref{eq:largestk2} evaluated for the parameters of the considered code is an integer or not). If a code attains this bound with equality but is not MRD, then it is called a \emph{quasi-MRD code}.
\end{remark}

In order to make use of the results in Sections~\ref{sec:bounds} and~\ref{sec:asy} we need the volume of the ball in the rank 
metric and its asymptotic estimates. The volume of the rank-metric ball of radius $r$ is given by
\begin{equation} \label{eq:rkballasy}
    \bbqrk{\F_{q^m}^n,r}:= \sum_{i=0}^r \qbin{n}{i}{q} \prod_{j=0}^{i-1}(q^{m}-q^{j})
\end{equation}
for any $0 \le r \le \min\{n,m\}$; see for example \cite{gabidulin}. 

As we are mainly interested in asymptotic results we need the following lemma, which states the well-known asymptotic estimates of the rank-metric ball as $q \to +\infty$, as $m \to +\infty$ and as~$n \to +\infty$.

\begin{lemma} \label{lem:rkballasy}
The following estimates hold.
\begin{itemize}
    \item[(i)] Let $0 \le r \le \min\{n,m\}$. We have $\bbqrk{\F_{q^m}^n,r} \sim q^{r(m+n-r)}$ as $q \to +\infty$.
    \item[(ii)] Let $0 \le r \le n$. We have $\bbqrk{\F_{q^m}^n,r} \sim
    \qbin{n}{r}{q}q^{rm}$ as $m \to +\infty$.
    \item[(iii)] Let $0 \le r \le m$. We have $\bbqrk{\F_{q^m}^n,r} \sim \qbin{m}{r}{q}q^{rn}$ as $n \to +\infty$.
\end{itemize}
\end{lemma}

\subsection{Nonlinear MRD Codes}

In this short subsection we investigate the density of (possibly nonlinear) MRD codes both as $q \to +\infty$ and $n \to +\infty$. 

Note that by Remark~\ref{rem:matrixiso} computing the asymptotic densities for $n \to +\infty$ and $m \to +\infty$ for nonlinear MRD codes in $\F_{q^m}^n$ are the same (up to transposition) and thus we will only treat the case where $n \to +\infty$ in this subsection.

In the following theorem we provide asymptotic results on the density function of nonlinear MRD codes both as $q \to +\infty$ and $n \to +\infty$.

\begin{theorem} \label{thm:nonlinmrdqn}
Let $2 \le d$ be an integer. 
\begin{itemize}
    \item[(i)] Let $2 \le m$ and $2 \le n$ be integers. We have 
    \begin{align*}
    \lim_{q \to +\infty}\delta_q^{\textnormal{rk}} (\F_{q^m}^n, q^{\max\{n,m\}(\min\{n,m\}-d+1)},0,d) =0.
    \end{align*}
    \item[(ii)] Let $2 \le m$ be an integer and suppose that $d \le m$. We have
    \begin{align*}
    \lim_{n \to +\infty}\delta_q^{\textnormal{rk}} (\F_{q^m}^n, q^{n(m-d+1)},0,d) =0.
    \end{align*}
\end{itemize}
Thus, nonlinear MRD codes are sparse both as $q \to +\infty$ and as $n \to +\infty$.
\end{theorem} 
\begin{proof}
From the estimate given in Lemma~\ref{lem:rkballasy} we get \begin{align*}
\lim_{q \to +\infty} \frac{q^{mn}}{\bbqrk{\F_{q^m}^n,d-1}q^{2\max\{n,m\}(\min\{n,m\}-d+1)}} = \lim_{q \to +\infty} \frac{1}{q^{(d-1)(m+n+2\max\{n,m\}-d+1)+mn}} = 0
\end{align*}
which by Theorem~\ref{thm:nonLinearLimq} proves the first result.

Similarly, we have
\begin{align*}
  \lim_{n \to +\infty} \frac{q^{mn}}{\bbqrk{\F_{q^m}^n,d-1}q^{2n(m-d+1)}} =  \lim_{n \to +\infty} \frac{1}{\qbin{m}{d-1}{q}q^{n(m-d+1)}} =0
\end{align*}
where we used the asymptotic estimate for $n \to +\infty$ in Lemma~\ref{lem:rkballasy}. The second statement in the theorem is then a consequence of Theorem~\ref{thm:nonLinearLimn}. 
\end{proof}

\subsection{(Sub)linear MRD Codes}
The problem of determining whether ($\F_q$-linear and $\F_{q^m}$-linear) MRD codes are dense or sparse has been studied before and we start this subsection by revisiting the results and their approaches that have been developed so far (see e.g.~\cite{antrobus2019maximal,byrne2020partition,gluesing2020sparseness,gruica2022common,gruica2022rank,neri2018genericity}), one of which is the application of the results presented in Section~\ref{sec:bounds}. We give a short recap of some of the results obtained with the other approaches in the language of this paper.

In~\cite{gruica2022common} it was shown that $[\F_{q^m}^n,q^{m(n-d+1)},1,d]$-MRD codes are sparse as $q \to +\infty$ unless $d=1$ or $n=d=2$. Together with a result from~\cite{antrobus2019maximal}, where they computed the exact asymptotic density for when $n=d=2$ using the theory of spectrum-free matrices, the density question for $\F_q$-linear MRD codes as $q\to +\infty$ is fully solved. Moreover, in~\cite{gluesing2020sparseness}, for the case where $m=n=d=3$, an exact asymptotic estimate as for the density function of $[\F_{q^m}^n,q^{m(n-d+1)},1,d]$-MRD codes $q \to +\infty$ was provided. This result was later generalized in~\cite{gruica2022rank} for any $m=n=d$ that are prime. 

Both the approach of~\cite{gluesing2020sparseness} and the one of~\cite{gruica2022rank} are based on the connection between full-rank MRD codes and semifields. 

Finally, using the Schwartz-Zippel Lemma, in~\cite{neri2018genericity} it was shown that $\F_{q^m}$-linear MRD codes are dense as $m \to +\infty$. For $\F_q$-linear MRD codes it has only been shown that they are \emph{not} dense as $m \to +\infty$ (see~\cite{antrobus2019maximal,byrne2020partition,gruica2022common}). Whether they are sparse or not is to this day still an open question.

As the density of MRD codes has already been investigated rather thoroughly, in this subsection we will just close the gap by discussing the remaining asymptotic density results. We start by looking at the (sub)linear case for the field size $q$ tending to infinity. As explained in Remark~\ref{rem:quasiMRD}, when $\ell s < n$, then the codes of maximum possible dimension are not necessarily MRD codes. For this reason, the following theorem is split into two parts.

\begin{theorem} \label{thm:sublinmrdq}
Let $n$, $\ell$ and $s$ be positive integers and let $2 \le d \le n$. 
\begin{itemize}
    \item[(i)] If $n \le \ell s$ then we have
    \begin{align*}
\lim_{q \to +\infty} \delta_q^\textnormal{rk}(\F_{q^{\ell s}}^n,q^{\ell s(n-d+1)},\ell,d) =
\begin{cases}
     1 \quad &\textnormal{ if $(d-1)(n-d+1) < \ell$,} \\
    0 \quad &\textnormal{ if $\ell < (d-1)(n-d+1)$.}
    \end{cases}
\end{align*}
Moreover, if $\ell = (d-1)(n-d+1)$ then $\lim_{q \to +\infty} \delta_q^\textnormal{rk}(\F_{q^{\ell s}}^n,q^{\ell s(n-d+1)},\ell,d) \le 1/2$.
    \item[(ii)] If $\ell s < n$ then we have
   \begin{align*}
\lim_{q \to +\infty} \delta_q^\textnormal{rk}(\F_{q^{\ell s}}^n,q^{\ell k},\ell,d) =
\begin{cases}
     1 \quad &\textnormal{ if $(d-1)(\ell s -d+1)+r < \ell$,} \\
    0 \quad &\textnormal{ if $\ell < (d-1)(\ell s -d+1)+r$.}
    \end{cases}
\end{align*}
where $k:=\lfloor n(\ell s-d+1)/\ell \rfloor$ and $r:=n(d-1)-\ell \lceil n(d-1)/\ell \rceil $.
Moreover, if $\ell = (d-1)(\ell s -d+1)+r$ then $\lim_{q \to +\infty} \delta_q^\textnormal{rk}(\F_{q^{\ell s}}^n,q^{\ell k},\ell,d) \le 1/2$.
\end{itemize}
\end{theorem}

\begin{proof}
For the first part of the theorem, note that we have
\begin{align*}
    \lim_{q \to +\infty} \frac{\bbqrk{\F_{q^{\ell s}}^n,d-1}q^{\ell s (n-d+1)}}{q^{\ell(ns+1)}}= \lim_{q \to +\infty} q^{(d-1)(n-d+1)-\ell} = 0
\end{align*}
if and only if $(d-1)(n-d+1) < \ell$. 
On the other hand, 
\begin{align*}
    \lim_{q \to +\infty} \frac{q^{\ell(ns+1)}}{\bbqrk{\F_{q^{\ell s}}^n,d-1}q^{\ell s (n-d+1)}}= \lim_{q \to +\infty} \frac{1}{q^{(d-1)(n-d+1)-\ell}} = 0
\end{align*}
if and only if $\ell < (d-1)(n-d+1)$.
Finally, the case $\ell = (d-1)(n-d+1)$ is another easy application of Theorem~\ref{thm:sublimq}.

To prove the second part of the theorem, we proceed analogously to the arguments for the first part, where we additionally use the fact that $k = \lfloor{ n(\ell s-d+1)}/{\ell}\rfloor = ns-\lceil n(d-1)/\ell \rceil$.
\end{proof}

We now turn to the question if $\mathbb{F}_{q^{\ell}}$-linear MRD codes in $\F_{q^{\ell s}}^n$ are dense or sparse as $\ell$, $s$ and~$n$ tend to infinity.

Recall that in~\cite{neri2018genericity} it was shown that $\F_{q^m}$-linear MRD codes are dense as $m \to +\infty$. The following result generalizes this fact, and shows that also $\F_{q^\ell}$-linear MRD codes in $\F_{q^{\ell s }}^n$ are dense as $\ell \to +\infty$.

\begin{theorem}
Let $q \in Q$ and suppose that $3 \le n$ and $1 \le \ell,d$ are integers with $1 \leq d \leq n$.
Then 
$$\lim_{\ell \to +\infty} \delta_q^\textnormal{rk}(\F_{q^{\ell s}}^n,q^{ \ell s(n-d+1)},\ell,d)= 1,$$
i.e., $\F_{q^\ell}$-linear MRD codes in $\F_{q^{\ell s }}^n$ are dense as $\ell \to +\infty$.
\end{theorem}
\begin{proof} 
By~\eqref{eq:rkballasy} we have
\begin{align*}
    \bbqrk{\F_{q^{\ell s}}^n,d-1} \sim \qbin{n}{d-1}{q}q^{(d-1)\ell s} \quad \textnormal{ as $\ell \to +\infty$.}
\end{align*}
Therefore $q^{\ell s(n-d+1)} \in o\left({q^{\ell (ns-1)}}/{\bbqrk{\F_{q^{\ell s}}^n,d-1}}\right)$ as $\ell \to +\infty$, and Theorem~\ref{thm:sublimell} concludes the proof.
\end{proof}

\begin{theorem} \label{thm:mrdlims}
Let $q \in Q$ and suppose that $3 \le n$ and  $1 \le \ell,d$ are integers with $2 \leq d \leq n$.
Then 
$$\limsup_{s \to +\infty} \delta_q^\textnormal{rk}(\F_{q^{\ell s}}^n,q^{\ell s(n-d+1)},\ell,d) \le \frac{q^\ell}{q^\ell+\qbin{n}{d-1}{q}} < 1.$$
\end{theorem}
\begin{proof} 
We will apply Theorem~\ref{thm:sublims} for proving the statement. First note that by~\eqref{eq:rkballasy} we have
\begin{align*}
    \bbqrk{\F_{q^{\ell s}}^n,d-1} \sim \qbin{n}{d-1}{q}q^{(d-1)\ell s} \quad \textnormal{ as $s \to +\infty$.}
\end{align*}
In particular, we have
\begin{align*}
    1- \frac{\bbqrk{\F_{q^{\ell s}}^n,d-1}}{q^{\ell(ns-s(n-d+1)+1)} + \bbqrk{\F_{q^{\ell s}}^n,d-1}} =  \frac{q^{\ell(s(d-1)+1)}}{q^{\ell(s(d-1)+1)} + \bbqrk{\F_{q^{\ell s}}^n,d-1}} \sim \frac{q^\ell}{q^\ell+\qbin{n}{d-1}{q}}
\end{align*}
as $s \to +\infty$, which proves the theorem.
\end{proof}

We now investigate the density of codes in $\F_{q^{\ell s}}^n$ as their vector length $n$ tends to infinity. Recall again that by Remark \ref{rem:quasiMRD}, in this setting, codes of the largest possible dimension are not necessarily MRD. 
\begin{theorem} \label{thm:mrdlimn}
Let $q \in Q$, let $1 \le \ell,s$ be integers and fix $2 \le d \le \ell s$. 
We have
$$\limsup_{n \to +\infty} \delta_q^\textnormal{rk}(\F_{q^{\ell s}}^n,q^{\ell k(n)},\ell,d) \le \frac{1}{1+\qbin{m}{d-1}{q}q^{-2\ell}} < 1,$$
where $k(n):=\lfloor{ n(\ell s-d+1)}/{\ell}\rfloor$ for all $2 \le n$.
\end{theorem}
\begin{proof} 
Applying Theorem \ref{thm:sublimn} and the asymptotic estimate as $n \to +\infty$ we get
\begin{align} \label{eq:samesame}
\limsup_{n \to +\infty} \delta_q^\textnormal{rk}(\F_{q^{\ell s}}^n,q^{\ell k(n)},\ell,d) \le \limsup_{n \to +\infty} \frac{1}{1+\qbin{m}{d-1}{q}q^{n(d-1)-\ell \lceil n(d-1)/\ell \rceil -\ell}},
\end{align}
where we used that $k(n) = \lfloor{ n(\ell s-d+1)}/{\ell}\rfloor = ns-\lceil n(d-1)/\ell \rceil$. Since $0 \le \ell \lceil n(d-1)/\ell \rceil -n(d-1) \le \ell$ we infer the statement in the theorem.
\end{proof}

Recall that by Remark~\ref{rem:matrixiso} there is a one-to-one correspondence between $\F_q$-linear MRD codes in $\F_{q^m}^n$, and $\F_q$-linear MRD codes in $\F_{q^n}^m$. Therefore, for $\ell=1$, the upper bound for the density as $m \to +\infty$ in Theorem~\ref{thm:mrdlims} and the upper bound for the density as $n \to +\infty$ in~\eqref{eq:samesame} are the same (up to exchanging $n$ and $m$).\footnote{Analogously for larger $\ell$ if $\ell$ divides $n$.} 

\section{Application III: Codes in the Sum-Rank Metric}\label{sec:SumRank}

In the scope of this paper, we only investigate sum-rank-metric codes $\mC$ in $\F_{q^m}^n$ that consist of block vectors $[x_1 | \dots | x_t]^\top \in \F_{q^m}^n$, where all blocks have equal size such that $x_i \in \F_{q^m}^{\eta}$ for all $i \in [t]$ and $n=\eta t$.
We define the sum-rank weight and sum-rank distance as follows.

\begin{definition}\label{def:sumrk}
Let $x =  [x_1 | \dots | x_t]^\top \in \F_{q^m}^n$. The ($t$-)\emph{sum-rank weight} of $x$ is $ \omega^{sr,t}(x)$ where~$ \wsr$ is the function defined as
    \begin{equation*}
        \wsr  :  \F_{q^m}^n \longrightarrow \mathbb{N}, \hspace{0.5em} x \mapsto \sum_{i=1}^{t}\wrk(x_{i}),
    \end{equation*}
where $\wrk(x_{i}) := \dim_{\F_{q}}\langle x_{i,1}, \dots ,x_{i,\eta} \rangle $  denotes the dimension of the $\F_{q}$-span of the entries of $x_{i}$.
For vectors $x,y \in \F_{q^m}^n$ the (t-)\emph{sum-rank distance} is $\dsr(x,y):=\wsr(x-y)$.
\end{definition}

Note that we clearly have that $\wrk(x_{i}) \leq \min\{m,\eta\}$ for $x_i \in \F_{q^m}^{\eta}$.
As usual, if $\mC \subseteq \F_{q^m}^n$ is an $\F_{q^{\ell}}$-linear subspace of dimension $k$ and minimum distance at least $d$ we say that $\mC$ is a $[\F_{q^m}^n,q^{\ell k},\ell,d]^{\textnormal{sr,t}}$-code, and if it is (possibly) nonlinear and has cardinality $S$, we say it is a $[\F_{q^m}^n,S,0,d]^{\textnormal{sr,t}}$-code.

\begin{remark} \label{rem:specialsumrkinst}
Note that if we set $\eta=1$ in Definition~\ref{def:sumrk}, then the sum-rank weight of a vector $x \in \F_{q^m}^n$ is the number of nonzero entries of $x$, and thus it reduces to the Hamming metric. On the other hand, if $t=1$ then it is easy to see that the sum-rank weight of a vector is the standard rank weight. Therefore, codes in the Hamming metric and codes in the rank metric can be seen as special instances of codes in the sum-rank metric.
\end{remark}

\begin{theorem} (Singleton-type bound; \text{see \cite[Theorem  3.2]{byrne2021Fundamental}}) \label{thm:sumrksingl}
Let $\mC \subseteq \F_{q^m}^n$ be a code in the sum-rank metric. Then the cardinality of $\mC$ is upper bounded by 
\begin{equation} \label{eq:Singletonlike_sr}
   |\mC| \leq q^{\max\{m, \eta \}(t\min\{m,\eta \}-d+1)}.
\end{equation}
\end{theorem}

A sum-rank-metric code is called \emph{MSRD} (\emph{maximum sum-rank distance}) code if its cardinality meets the Singleton-type bound in Theorem~\ref{thm:sumrksingl} with equality.

\begin{remark} \label{rem:quasiMSRD}
Let $\ell$ be a divisor of $m$ and let $s=[\F_{q^{m}}\colon \F_{q^{\ell}}]$. The $\F_q$-dimension of any $\F_{q^\ell}$-linear sum-rank-metric code $\mC \subseteq \F_{q^{\ell s}}^n$ has to be divisible by $\ell$. In particular, if  $\dim_{\F_{q^{\ell}}}(\mC)=k$, where $0 \le k \le ns$, then from the bound in Equation \eqref{eq:Singletonlike_sr}, we have
\begin{align} \label{eq:largestk_sr}
    \ell k \le \max\{\ell s, \eta \}(t\min\{\ell s ,\eta \}-d+1),
\end{align}
where $d=\drk(\mC)$. If $\ell s  \le \eta$, then the largest integer $k$ that satisfies \eqref{eq:largestk_sr} is 
\begin{align*}  
   k^{*}:=\left\lfloor \frac{\eta(\ell s t -d+1)}{\ell}\right\rfloor.
\end{align*}
Clearly, a code $\mC$ with $\dim_{\F_{q^{\ell}}}(\mC)= k^{*}$ is not necessarily MSRD. A code $\mC$ with $\dim_{\F_{q^{\ell}}}(\mC)= k^{*}$ that does not attain the Singleton bound with equality is called a \emph{quasi-MSRD code}.
\end{remark}

In what follows, we define the set of all possible partitions of a number $r$ into exactly $t$ parts by $U_{r}$, where each part is upper bounded by $\min\{m,\eta\}$. This set of partitions will be needed for giving the volume of the sum-rank-metric ball.

\begin{notation}
Let $2 \leq n$ and $0 \le r \le t  \min\{m,\eta\}$. Suppose that $1 \leq t$ is an integer and let $n = \eta t$. We define 
\begin{equation*} 
    U_{r}:= \left\{u=(u_{1}, \dots, u_{t}) \in \mathbb{N}_{0}^{t} \bigm|\, \sum_{i=1}^{t}u_{i}=r, u_{i} \leq \min\{m,\eta\} \textnormal{ for all } i \in [t] \right\}.
\end{equation*}
\end{notation}

The following lemma gives a closed formula for the volume of the sum-rank-metric ball of a given radius.
\begin{lemma}\cite{byrne2021Fundamental}  
    Let $0 \leq r \le t\min\{m,\eta\}$. We have 
\begin{equation*} 
    \bbqsr{\F_{q^m}^n, r} :=  \displaystyle\sum_{h=0}^r \displaystyle\sum_{u \in U_{h}} \displaystyle\prod_{i=1}^t\dstirling{\eta}{u_{i}}_{q} \prod_{j=0}^{u_{i}-1}(q^{m}-q^{j}).
\end{equation*}
\end{lemma}
The following lemma gives the asymptotic estimates of the volume of the sum-rank-metric ball as $q$ goes to infinity (and the other parameters are treated as constants).
 
\begin{lemma} \label{lem:srkballasy}
Let $t$ be a divisor of $n$, let $\eta = \frac{n}{t}$ and let $0 \le r \le t\min\{m,\eta\}$. We have
\begin{align*}
     \bbqsr{\F_{q^m}^n, r} \sim \binom{t}{\tilde{z}} q^{\frac{\tilde{z}^2}{t}-\tilde{z}+r(m+\eta -\frac{r}{t})}
 \quad \textnormal{ as $q \to +\infty$,}
\end{align*}
where $\tilde{z} \equiv r \pmod t$.
\end{lemma}
\begin{proof}

Note that for any $u = (u_1,\dots,u_t) \in U_{h}$ we have
\begin{align*}
    \prod_{i=1}^t \qbin{\eta}{u_i}{q} \prod_{j=0}^{u_i-1} (q^m-q^j) \sim q^{\sum_{i=1}^t u_i(m+\eta-u_i)} = q^{h(\eta +m)-\sum_{i=1}^t u_{i}^{2}} \quad \textnormal{ as } q \to +\infty.
\end{align*}
Let $z \equiv h \pmod t$. Using Lagrange multipliers one sees that the maximum of the real-valued function $f:\R^t \longrightarrow \R, \, f  :  (x_1, \dots, x_t) \mapsto -\sum_{i=1}^t x_{i}^{2}$ in the region of $\R^t$ constrained by 
\begin{align*}
    \sum_{i=1}^t x_i = h, \quad 0 \le x_i \le \eta \textnormal{ for all } i \in \{1, \dots, t\}
\end{align*}
is attained for $x^{*} \in \R^t$ with $x^{*}_{1} = \dots = x^{*}_{t}= \frac{h}{t}$. Note that $f$ is concave in all $t$ variables, hence the integer partition $u \in U_{h}$ where $f$ reaches its maximum is $u^{*}= (u^{*}_{1},\dots,u^{*}_{t})$ with $u^{*}_{1} = \dots = u^{*}_{i_{z}}= \lfloor \frac{h}{t} \rfloor +1$ and $u^{*}_{k}=\lfloor \frac{h}{t} \rfloor$ where $k \in [t]\setminus \{i_{1} , \dots ,i_{z} \}$.
This gives
\begin{align*}
    \max_{u \in U_{h}}\left( -\sum_{i=1}^t u_{i}^{2} \right) &= -z(\left\lfloor \frac{h}{t} \right\rfloor+1)^{2}-(t-z)\left\lfloor \frac{h}{t} \right\rfloor^{2}\\
    &=- 2\left\lfloor \frac{h}{t} \right\rfloor z- \left\lfloor \frac{h}{t} \right\rfloor^{2}t-z \\
    &= -\left\lfloor \frac{h}{t} \right\rfloor(z+h)-z\\
    &= -\frac{(h-z)(h+z)}{t}-z = -\frac{h^{2}-z^{2}}{t}-z.
\end{align*}
Hence
\begin{align*}
    \displaystyle\sum_{u \in U_{h}} \displaystyle\prod_{i=1}^t\dstirling{n_{i}}{u_{i}}_{q} \prod_{j=0}^{u_{i}-1}(q^{m}-q^{j}) \sim \binom{t}{z}q^{h(m + \eta)-z-h^{2}/t + z^{2}/t} \quad \textnormal{ as $q \to +\infty$.}
\end{align*}
Write $p(h) = h(mt+n-h)$ as a polynomial in $h$ with roots at $h=0$, $h=mt+n$ and a maximum at $h^{*}= \frac{mt+n}{2}$. Note that $p$ is monotonically increasing on  $\left[0, \frac{mt+n}{2}\right)$. For $h \le r \le t\min\{\eta,m\}$ we have
$$h  \le r \le \frac{mt+n}{2},$$
from which we can conclude that
$$ \bbqsr{\F_{q^m}^n, r} \sim \binom{t}{\tilde{z}} q^{\frac{\tilde{z}}{t}-\tilde{z}+r(m+\eta -\frac{r}{t}) },$$
where $\tilde{z} \equiv r \pmod t$.
\end{proof}
As we saw in Remark \ref{rem:specialsumrkinst}, both the Hamming and the rank metric are special instances of the sum-rank metric. By analyzing the asymptotic densities of (sub)linear codes in these two metrics, we can infer that their behavior differs, especially when looking at $q \to +\infty$. This interesting fact motivates us to generalize and to analyze the asymptotic density as $q$ tends to infinity in the sum-rank metric.

\subsection{Nonlinear MSRD Codes}

We determine the sparsity of (possibly) nonlinear MSRD codes as $q \to +\infty$. 

\begin{theorem} \label{thm:nonlinearnSr}
Let $2 \leq d$ and $1 \leq m,\eta$ be integers, let $t$ be the number of blocks in the sum-rank metric and let $n=\eta t$.
Then we have
$$\lim_{q \to +\infty} \delta_q^{\textnormal{sr,t}}(\F_{q^{m}}^n,q^{\max\{m,\eta\}(t\min\{m,\eta\}-d+1)},0,d)= 0,$$
i.e., nonlinear MSRD codes are sparse as $q \to +\infty$.

\begin{proof}
By Lemma \ref{lem:srkballasy} we have
\begin{align*}
0 \leq \lim_{q \to +\infty} \frac{q^{mn}}{\textbf{v}_{q}^{\textnormal{sr,t}}(\F_{q^m}^n,d-1)q^{2\max\{m,\eta\}(t\min\{m,\eta\}-d+1)}}\leq \lim_{q \to +\infty} \frac{1}{q^{(d-1)(3\max\{m,\eta\})}}=0
\end{align*}
and Theorem~\ref{thm:nonLinearLimq} proves the statement.
\end{proof}
\end{theorem}

\subsection{(Sub)linear MSRD Codes}

Note that if $\ell s < \eta$, then the codes of maximum possible dimension are not necessarily MSRD codes (see also Remark \ref{rem:quasiMSRD}). Due to this reason we consider the two cases $\ell s < \eta$ and $\eta \le \ell s$ separately.

\begin{theorem} \label{thm:asympsrq}
Let $\eta, \ell , s, t ,d$ be positive integers such that $1<t<n$ and $2 \le d  \le t\min\{m,\eta\}$. Let $n = \eta t$ and $\tilde{z} \equiv d-1 \pmod t$. Then define $\theta := (d-1)\left(\min\{m, \eta \} - \frac{d-1}{t} \right) +\frac{\tilde{z}^{2}}{t}-\tilde{z}$.
\begin{itemize}
   \item[(i)]
   If $\eta \le m$ then we have
   \begin{align*}
     \lim_{q \to +\infty}  \delta_q^{\textnormal{sr,t}}(\F_{q^{m}}^n,q^{m(n-d+1)},\ell,d)  = \begin{cases}
 1 \quad  &\textnormal{if }  \theta < \ell,  \\
0 \quad  &\textnormal{if }  \ell < \theta.
      \end{cases}
\end{align*}
   Moreover if $\theta = \ell$ then $  \lim_{q \to +\infty}  \delta_q^{\textnormal{sr,t}}(\F_{q^{m}}^n,q^{m(n-d+1)},\ell,d)  \leq \frac{1}{1+\binom{t}{\tilde{z}}}.$
 \item[(ii)] If $m < \eta$ then we have
  \begin{align*}
     \lim_{q \to +\infty}  \delta_q^{\textnormal{sr,t}}(\F_{q^{m}}^n,q^{\ell k},\ell,d) = \begin{cases}
 1 \quad  &\textnormal{if }  \theta -r  < \ell,   \\
  0 \quad  &\textnormal{if }  \ell < \theta -r.
       \end{cases}
\end{align*}
Moreover, if $\theta -r  = \ell$ then
  $\lim_{q \to +\infty}  \delta_q^\textnormal{sr,t}(\F_{q^{m}}^n,q^{\ell k},\ell,d) \leq \frac{1}{1+\binom{t}{\tilde{z}}}$ ,  
  
  where $k= \left\lfloor \frac{\eta(m t -d+1)}{\ell} \right\rfloor$ and $r = \ell \left( \left\lceil \frac{\eta(d-1)}{\ell} \right\rceil - \frac{\eta(d-1)}{\ell}\right)$.
\end{itemize}
\begin{proof}
By Lemma \ref{lem:srkballasy} we have for the first part of the theorem
\begin{align*}
\lim_{q \to +\infty} \frac{ \bbqsr{\F_{q^{m}}^n, d-1}q^{m(\eta t -d+1)}}{q^{\ell(ns+1)}} = \lim_{q \to +  \infty}\frac{\binom{t}{\tilde{z}}q^{(d-1)(\eta -\frac{d-1}{t})+\frac{\tilde{z}^{2}}{t}-\tilde{z}}}{q^{\ell}} =0,
\end{align*}
if and only if $\theta < \ell$. On the other hand
\begin{align*}
\lim_{q \to +\infty} \frac{q^{\ell(ns+1)}}{\bbqsr{\F_{q^{m}}^n, d-1}q^{m(\eta t -d+1)}} = \lim_{q \to +  \infty}\frac{1}{\binom{t}{\tilde{z}}q^{(d-1)(\eta -\frac{d-1}{t})+\frac{\tilde{z}^{2}}{t}-\tilde{z}-\ell}} =0,
\end{align*}
if and only if $\ell < \theta$. Finally, for $\ell = \theta$ we obtain 
\begin{align*}
\lim_{q \to +\infty} \frac{ \bbqsr{\F_{q^{m}}^n, d-1}q^{m(\eta t -d+1)}}{q^{\ell(ns+1)}} = \binom{t}{\tilde{z}},
\end{align*}
and the third case is another easy consequence of Theorem~\ref{thm:sublimq}. 

To prove the second part of the theorem, we proceed analogously to the arguments for the first part, where we additionally use the fact that $k = \lfloor{ \eta(\ell s t-d+1)}/{\ell}\rfloor = ns-\lceil \eta(d-1)/\ell \rceil$.
\end{proof}
\end{theorem}

As the density of $\F_{q}$-linear codes in the Hamming metric and rank metric has been studied pretty thoroughly, we will conclude this section by taking a closer look at their counterparts in the sum-rank metric. 

\begin{theorem} \label{thm: Fqlinearsr}
Let $\eta \le m$ and $2 \le d \le n$ be integers. We have \begin{align*}
    \lim_{q \to +\infty}\delta_q^\textnormal{sr,t}(\F_{q^{m}}^n,q^{m( n-d+1)},1,d) = \begin{cases}
         1 \quad &\textnormal{ if $(d-1)(n-d+1) < t$,} \\
         0 \quad &\textnormal{ if $t+t^2/4 < (d-1)(n-d+1)$.}
    \end{cases}
\end{align*}

\end{theorem}
\begin{proof}
Note that, using simple methods from Calculus, one can show that $$-t/4 \le \frac{(\tilde{z})^2}{t}-\tilde{z} \le 0,$$
where $\tilde{z} \equiv d-1 \pmod t$. Combining this with Theorem~\ref{thm:asympsrq}, where we set $\ell=1$ and use the fact that $\eta \le m$ concludes the proof.
\end{proof}
From the assumption $2 \le d \le n$ in Theorem \ref{thm: Fqlinearsr} we get
$$(d-1)(n-d+1) \in \left\{ (n-1), \dots , \left\lfloor \frac{n^{2}}{4}\right\rfloor \right\}.$$
This gives us a further characterization (see Corollary \ref{cor:boundsForeta} and Figure \ref{figure:srbounds} below) of sum-rank-metric codes from which we can extract all values of block length $\eta$, where we observe a sparse behaviour as $q \to +\infty$.
\begin{corollary} \label{cor:boundsForeta}
    Let $\eta \le m$ and $2 \le d \le n$ be integers. We have \begin{align*}
    \lim_{q \to +\infty}\delta_q^\textnormal{sr,t}(\F_{q^{m}}^n,q^{m( n-d+1)},1,d) = \begin{cases}
         1 \quad &\textnormal{ if $\eta < 2/\sqrt{t}$}\\
         0 \quad &\textnormal{ if $\frac{(t+2)^{2}}{4t} < \eta$.}
    \end{cases}
\end{align*}
\end{corollary}
\begin{figure}[hbt!]
  \centering
  \begin{tikzpicture}
   \begin{axis}[ axis lines=middle, xmin=1,xmax=10,
    ymin=1,ymax=4,
    width=\textwidth,
    height=0.4\textwidth,
   legend style={at={(1,0.2)},anchor=south
   east},
    xlabel=$t$, ylabel=$\eta (t)$,restrict y to domain=0:100, ]
        \addplot[color=blue, domain=1:10,samples=301, unbounded coords=discard] {(2/sqrt(x))}; 
         \addlegendentry{$\frac{2}{\sqrt{t}}$}
           \addplot[color=red, domain=1:10,samples=301, unbounded coords=discard] {((x+2)^2/(4*x))};
            \addlegendentry{$\frac{(t+2)^{2}}{4t}$}        
     \addplot[color=green, only marks,
    mark=halfcircle*,
    mark size=0.8pt] coordinates {
               (1,1)
    (2,1)
   (3,1)};
     \addplot[color=green, only marks,
    mark=halfcircle*,
    mark size=0.8pt] coordinates {
 (1,3)
      (1,4)
       (2,3)
      (2,4)
      (3,3)
      (3,4)
      (4,3)
      (4,4)
      (5,3)
      (5,4)
      (6,3)
      (6,4)
      (7,3)
      (7,4)
      (8,4)
          (9,4)
          (10,4)
      };
     \addplot[color=black, only marks,
    mark=halfcircle*,
    mark size=0.8pt] coordinates {
     (4,1)
     (5,1)
     (6,1)
      (7,1)
       (8,1)
        (9,1)
         (10,1)
         (1,2)
          (2,2)
           (3,2)
              (4,2)
     (5,2)
     (6,2)
      (7,2)
       (8,2)
            (8,3)
                 (9,3)
                      (10,3)
        (9,2)
         (10,2)
      };
   
               \end{axis}
  \end{tikzpicture}
    \caption{Block size $\eta$ depending on $t$ to fully characterize the asymptotic density of $\F_{q}$-linear MSRD codes as $q \to +\infty$. The green dots above the red line represent sets of parameters for which MSRD codes are sparse in $\F_{q^{m}}^{\eta t}$, whereas those below the blue line indicate a dense behaviour. The black dots represent sets of parameters for which the asymptotic behaviour of the density is unknown, or as in the case of $t=1$ and $\eta =2$ can not be classified as asymptotically dense or sparse.}
      \label{figure:srbounds}
  \end{figure}
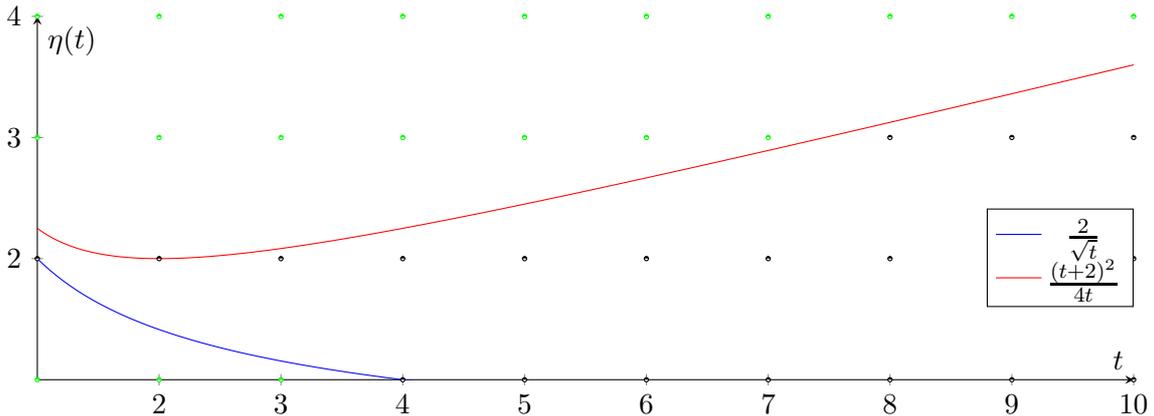
\begin{remark}\label{rem:comparison}
Recall from Remark~\ref{rem:specialsumrkinst} that codes in the Hamming metric and codes in the rank metric are subfamilies of sum-rank-metric codes. It is interesting to observe that these two classes of codes behave very differently with respect to density considerations as $q \to +\infty$. More explicitly, in Section~\ref{sec:Hamming} we saw that $\F_q$-linear MDS codes are dense as $q \to +\infty$. This is in stark contrast with the behavior of $\F_q$-linear MRD codes: they are sparse as $q \to +\infty$ (see Section~\ref{sec:Rank}). Even though there does not seem to be an obvious way to fully characterize for which values of $\eta, t$ and $d$ we have sparsity or density (or non-density) of the corresponding codes, experimental results strongly indicate that, in general, sum-rank-metric codes behave similarly to standard rank-metric codes, with an exception when $\eta=1$ (which is the case of Hamming-metric codes). Therefore MDS codes behave rather \emph{atypical} with respect to density considerations. We collect some examples of the asymptotic behavior of the density of sum-rank-metric codes in Table~\ref{table_sumrk}.
\begin{table}[h!]
    \centering
    \renewcommand\arraystretch{1.2}
 \begin{tabular}{|c|c|c|c|c|c|} 
 \hline
 \;$\eta$ \;&\; $t$ \;& \; $d$ \; & dense  & \emph{not} dense & sparse \\\noalign{\global\arrayrulewidth 1.8pt}
    \hline
    \noalign{\global\arrayrulewidth0.4pt}
 1 & 10 & 5 & $\checkmark$ & & \\
 \hline 
  $2$ & $\ge 1$ & 2 &  & $\checkmark$ &  \\
  \hline
 $\ge 2$ & 10 & 5 &  & & $\checkmark$ \\
  \hline
   $3$ & $\ge 1$ & 3 &  & & $\checkmark$ \\
  \hline
\end{tabular}
\caption{Examples illustrating the asymptotic behavior of the density of $\F_q$-linear sum-rank-metric codes as $q \to +\infty$.}
\label{table_sumrk}
\end{table}
\end{remark}

\FloatBarrier
\subsection*{Conclusion and Outlook}

We developed a general framework for determining densities of codes over finite fields, equipped with an arbitrary translation-invariant metric, with any linearity degree. This unifies several previous results (which were mostly derived with different mathematical tools) on densities of linear and nonlinear codes in the Hamming, rank and sum-rank metric. Moreover, we established new results for codes in these three metrics, answering many of the open question in this field.

One general observation is that nonlinear codes achieving the Gilbert-Varshamov bound in any metric, or the Singleton-type bound in the sum-rank metric (including the Hamming and rank metric as special cases), are always sparse, for both the field size or the length going to infinity. However, in the (sub)linear case, the metric, the linearity degree and the parameter going to infinity make a difference. In some cases we get sparsity, in some density and in some cases we can only derive an upper bound on the asymptotic density, showing that the code family is not dense (but we do not know if they are sparse). 

As a particular case of interest we studied the asymptotic behavior of MSRD codes for $q\rightarrow +\infty$, since MSRD codes generalize both the Hamming, as well as the rank metric. We derived bounds on the number of blocks for $\F_{q}$-linear MSRD codes to be dense or sparse. As explained in Remark \ref{rem:comparison}, the only case where MSRD codes are dense is when they are MDS codes. In all other cases we either get sparsity or an upper bound below $1$ for the density.

There are still many open questions for future work. E.g., in all cases where we obtained upper bounds below $1$ for the density, it is not clear if this bound is sharp or if there are sharper bounds, or if these code families are even sparse. Furthermore, we only presented the asymptotic densities of MSRD codes as the field size $q$ tends to infinity. With similar techniques we can derive results for the other parameters going to infinity, however, it was out of the scope of this paper to include them here; in particular, since they become more difficult
due to the vast number of variable parameters. Moreover, we restricted ourselves to MSRD codes where each block has the same length, and it remains an open problem to study the densities of such codes with different block lengths. Lastly, we will extend our framework to other metric spaces and study the densities of codes over finite rings, equipped with the Hamming or the Lee metric. 

\bigskip

\bibliographystyle{amsplain}
\bibliography{ourbib}

\providecommand{\bysame}{\leavevmode\hbox to3em{\hrulefill}\thinspace}
\providecommand{\MR}{\relax\ifhmode\unskip\space\fi MR }
\providecommand{\MRhref}[2]{%
  \href{http://www.ams.org/mathscinet-getitem?mr=#1}{#2}
}
\providecommand{\href}[2]{#2}
\begin{thebibliography}{10}

\bibitem{antrobus2019maximal}
J.~Antrobus and H.~Gluesing-Luerssen, \emph{Maximal {F}errers diagram codes:
  {C}onstructions and genericity considerations}, IEEE Transactions on
  Information Theory \textbf{65} (2019), no.~10, 6204--6223.

\bibitem{apostol2013introduction}
T.~M. Apostol, \emph{Introduction to {A}nalytic {N}umber {T}heory}, Springer
  Science \& Business Media, 2013.

\bibitem{baldi2021restricted}
M.~Baldi, M.~Battaglioni, F.~Chiaraluce, A.-L. Horlemann-Trautmann,
  E.~Persichetti, P.~Santini, and V.~Weger, \emph{A new path to code-based
  signatures via identification schemes with restricted errors},  (2021).

\bibitem{ball2020additive}
S.~Ball, G.~Gamboa, and M.~Lavrauw, \emph{On additive mds codes over small
  fields}, arXiv preprint arXiv:2012.06183 (2020).

\bibitem{barg2002random}
A.~Barg and G.~D. Forney, \emph{Random codes: Minimum distances and error
  exponents}, IEEE Transactions on Information Theory \textbf{48} (2002),
  no.~9, 2568--2573.

\bibitem{byrne2021Fundamental}
E.~Byrne, H.~Gluesing-Luerssen, and A.~Ravagnani, \emph{Fundamental properties
  of sum-rank-metric codes}, IEEE Transactions on Information Theory
  \textbf{67} (2021), no.~10, 6456--6475.

\bibitem{byrne2020partition}
E.~Byrne and A.~Ravagnani, \emph{Partition-balanced families of codes and
  asymptotic enumeration in coding theory}, Journal of Combinatorial Theory,
  Series A \textbf{171} (2020).

\bibitem{CVE}
P.-L. Cayrel, P.~V{\'e}ron, and S.~M. El~Yousfi~Alaoui, \emph{A zero-knowledge
  identification scheme based on the q-ary syndrome decoding problem}, Selected
  Areas in Cryptography (Berlin, Heidelberg), Springer Berlin Heidelberg, 2011,
  pp.~171--186.

\bibitem{de1981asymptotic}
N.~G. De~Bruijn, \emph{Asymptotic {M}ethods in {A}nalysis}, vol.~4, Courier
  Corporation, 1981.

\bibitem{delsarte1978bilinear}
P.~Delsarte, \emph{Bilinear forms over a finite field, with applications to
  coding theory}, Journal of Combinatorial Theory, Series A \textbf{25} (1978),
  no.~3, 226--241.

\bibitem{gabidulin}
E.~M. Gabidulin, \emph{Theory of codes with maximum rank distance}, Problemy
  Peredachi Informatsii \textbf{21} (1985), no.~1, 3--16.

\bibitem{gilbert1952comparison}
E.~N. Gilbert, \emph{A comparison of signalling alphabets}, The Bell system
  technical journal \textbf{31} (1952), no.~3, 504--522.

\bibitem{gluesing2020sparseness}
H.~Gluesing-Luerssen, \emph{On the sparseness of certain linear {MRD} codes},
  Linear Algebra and its Applications \textbf{596} (2020), 145--168.

\bibitem{gorla2018codes}
E.~Gorla and A.~Ravagnani, \emph{Codes endowed with the rank metric}, Network
  Coding and Subspace Designs, Springer, 2018, pp.~3--23.

\bibitem{gruica2021typical}
A.~Gruica and A.~Ravagnani, \emph{The typical non-linear code over large
  alphabets}, 2021 IEEE Information Theory Workshop (ITW), IEEE, 2021,
  pp.~1--6.

\bibitem{gruica2022common}
\bysame, \emph{Common complements of linear subspaces and the sparseness of
  {MRD} codes}, SIAM Journal on Applied Algebra and Geometry \textbf{6} (2022),
  no.~2, 79--110.

\bibitem{gruica2022rank}
A.~Gruica, A.~Ravagnani, J.~Sheekey, and F.~Zullo, \emph{Rank-metric codes,
  semifields, and the average critical problem}, arXiv preprint
  arXiv:2201.07193 (2022).

\bibitem{neri2018genericity}
A.~Neri, A.-L. Horlemann-Trautmann, T.~Randrianarisoa, and J.~Rosenthal,
  \emph{On the genericity of maximum rank distance and {G}abidulin codes},
  Designs, Codes and Cryptography \textbf{86} (2018), no.~2, 341--363.

\bibitem{ott2021bounds}
C.~Ott, S.~Puchinger, and M.~Bossert, \emph{Bounds and genericity of
  sum-rank-metric codes}, 2021 XVII International Symposium" Problems of
  Redundancy in Information and Control Systems"(REDUNDANCY), IEEE, 2021,
  pp.~119--124.

\bibitem{reed1960polynomial}
I.~S. Reed and G.~Solomon, \emph{Polynomial codes over certain finite fields},
  Journal of the society for industrial and applied mathematics \textbf{8}
  (1960), no.~2, 300--304.

\bibitem{shannon48}
C.~E. Shannon, \emph{A mathematical theory of communication}, The Bell System
  Technical Journal \textbf{27} (1948), no.~3, 379--423.

\bibitem{varshamov1957estimate}
R.~R. Varshamov, \emph{Estimate of the number of signals in error correcting
  codes}, Docklady Akad. Nauk, SSSR \textbf{117} (1957), 739--741.

\bibitem{Vron2009ImprovedIS}
P.~V{\'e}ron, \emph{Improved identification schemes based on error-correcting
  codes}, Applicable Algebra in Engineering, Communication and Computing
  \textbf{8} (2009), 57--69.

\end{thebibliography}

\end{document}